\pdfoutput=1

\documentclass[11pt]{article}

\usepackage{epsfig}
\usepackage[section]{algorithm}
\usepackage{algorithmic}
\usepackage{anysize}
\usepackage{amsmath}
\usepackage{amssymb}
\usepackage{color}
\usepackage{todonotes}
\usepackage{hyperref}
\hypersetup{colorlinks=true, urlcolor=blue}
\usepackage{thmtools}
\usepackage{nameref,cleveref}

\def\qed{\hbox{\rlap{$\sqcap$}$\sqcup$}}
\newenvironment{proof}{\par\noindent{\bf Proof:}}{\mbox{}\hfill$\qed$\\}

\newtheorem{theorem}{Theorem}[section]
\newtheorem{lemma}[theorem]{Lemma}

\newtheorem{corollary}[theorem]{Corollary}

\newtheorem{defi}{Definition}
\newenvironment{definition}{\begin{defi}\rm}{\end{defi}}

\newcounter{rem}
\newcommand{\ignore}[1]{ }

\begin{document}
\title{Price of Anarchy in Networks with Heterogeneous Latency Functions}
\author{Sanjiv Kapoor\thanks{%Computer Science,IIT, Chicago, IL, 
e-mail: kapoor@iit.edu} \  and Junghwan Shin\thanks{%Computer Science, IIT, Chicago, IL, 
 e-mail: byulpyo@gmail.com}}
\date{}
\maketitle

\begin{abstract}

We address the performance of selfish network routing  in multi-commodity  flows  where the latency or delay function on edges 
%termed as heterogeneous delay functions,
is dependent on the flow of individual commodities, rather than on the aggregate flow. An application of this study is the  analysis of a network with differentiated traffic, i.e., in transportation networks where there are multiple types of traffic and in networks where traffic is prioritized according to type classification. We consider the
inefficiency of equilibrium in this model and
provide {\em price of anarchy} bounds for networks with $k$ (types of) commodities where each link is
associated with  heterogeneous polynomial delays, i.e., commodity $i$ on edge $e$ faces delay
specified by
$h^e_i(f_1(e), f_2(e), \ldots, f_k(e)) $
where %$f(e)$ is the flow vector and
$f_i(e)$ is the flow of the $i$th commodity through edge $e$ and $h^e_i()$ a polynomial delay function applicable to the $i$th commodity.
We consider both atomic
and non-atomic flows and show bounds on the price of anarchy that depend on the relative impact of each type of traffic on the edge delay where the  delay functions  are  polynomials of degree $\theta$, e.g., $\sum_i a_i f_i(e)^\theta$.
The price of anarchy is unbounded for arbitrary polynomials.

For networks with {\em decomposable delay functions} where the delay is the same for all
commodities using the edge, i.e., delays on edge $e$ are defined by $h^e(f_1(e), f_2(e), \ldots, f_k(e))$,
we show improved bounds on the price of anarchy, for both non-atomic and atomic flows.

The results illustrate that the inefficiency of selfish routing worsens in the case of heterogeneous delays as compared to the standard delay
functions that do not consider type differentiation.
\end{abstract}

\section{Introduction}
A typical road network serves multiple types of traffic and each type has a different impact on the delay experienced on a link.  
The same is also true for computer network routing where traffic of different types may be  generated via priority mechanisms. The common theme in these networks is the lack of a centralized control, and thus each source-destination pair  chooses routes based on optimizing an objective. This problem can be modeled using game theory where each agent or player selfishly chooses its route, with Nash equilibrium being achieved when all the players have an optimum route for their own traffic, given the routes of the other players as fixed. 
There are two categories of games considered, the non-atomic routing game, where each player controls a negligible amount of traffic and the atomic routing game, where each user can  utilize only one routing path to satisfy its requirement.

These routing games can be formalized as weighted congestion games on a network, where  traffic requirements of  multiple types of network users are to be routed between corresponding source and destination pairs. Each edge has  a delay (latency) function that is dependent on the traffic of each user type present on that edge.  To satisfy their traffic  requirements, users determine their own routes using a selfish routing strategy (one that minimizes the delay faced on the paths chosen by the individual user).
% heterogeneous traffic in networks where users utilize autonomous selfish routing strategies.
A centralized  optimal routing solution would optimize user routes with respect to a measure of social benefit, the typical objective being the sum total of delays faced by the users; however the selfish strategic routing game introduces inefficiencies since Nash equilibrium solutions, being optimal for each individual player given the strategies of other players, do not necessarily optimize the social benefit. 
Our focus of study is 
%the problem of determining
the inefficiency of  network routing games that consider multiple types of traffic.

% as distinct from a the routing determined for each user by a central authority. 
% This gives rise to a network routing game with heterogeneous delays where the Nash equilibrium of the game provides routes for each user. 
% the term heterogeneous reflecting that the delay is impacted differently by player types.

The degradation of network performance that results from competitive selfish routing, as compared to a global optimum solution, and measured with respect to the total latency or delay of the flow
has been quantified in \cite{KP1999,RT2002} where they considered the {\em price of anarchy} (PoA).  In the flow routing context,  this is
defined to be  the ratio of the total latency of the Nash equilibrium routing to the total latency of the optimum routing.
The bounds that have been obtained in these previous works apply to  latency or delay on edges that are  a function of the aggregate flow on the edge.
In this paper we consider an important generalization that models the above mentioned scenario of heterogeneous traffic on $k$-commodity networks: in this model, edges are associated
with  delay functions that are dependent on the type of the commodity, which we  term as {\em heterogeneous} delay functions. An example is a polynomial delay of the form 
$\sum_i a_i(e) f^\theta_i(e) + c(e)$
% $a_1(e) f^{\theta}_1(e) + a_2(e) f^{\theta}_2(e) + \ldots + a_k(e) f^{\theta}_k(e) + c(e)$ 
with positive co-efficients and non-negative constant, where $f_i(e)$
represents the traffic contribution of the $i$th commodity on edge $e$, and $\theta \in \mathbb{N}$ the degree
of the polynomial.
%The co-efficients and constants 
%$a_1(e), a_2(e), \ldots, a_k(e)$ and $c(e)$ 
%are non-negative integers.
%thus bolstering the defense of {\em net-neutrality} and motivating further  similar studies.

In general, we model the total delay in heterogeneous networks using the following total (heterogeneous) delay function:
$\sum_{i,e}f_i(e)\Phi^i_e(f_1,\ldots,f_k) = \sum_{i,e}f_i(e)h^e_i(f_1(e) , f_2(e), \ldots, f_k(e))  $
where $f_i(e)$ is the flow function and $h^e_i : \mathbb{R}^k \rightarrow \mathbb{R}^+$  a convex, monotonically increasing, polynomial  delay function for
commodity $i$ on edge $e$, respectively.
Note that the distinctive property of this delay function is that the domain is multi-dimensional.
%is termed as a {\em heterogeneous} delay function.
If the delay function, $\Phi^i_e(f(e))$, is expressible as
$h^e(f_1(e) , f_2(e), \ldots, f_k(e)) $, 
i.e., the delay function on the edge
is independent of the commodity,
then we call this delay function {\em decomposable}.
This is  an interesting  class of functions, since it addresses the case when the delay on an edge is the same as experienced by each of the multiple types of users. We also consider  {\em uniform} delay functions where the delay on every edge of the network is the same function of the flow; however, the constant term may differ.%dependency on the flow vector 
% is the same  across all of the edges (the delay is still dependent on the flows through the edge though).
%% Further, we define {\em uniform} delay function as decomposable delay functions of the
% form $h(f_1, f_2, \ldots, f_k) $%+a(e)$
% , where the delay function %dependency on the flow vector 
% is the same across all of the edges (the delay is
% still dependent on the flows through the edge though).

The generalized model we consider, has applications in  transport networks utilized by heterogeneous modes of transport and studied by Dafermos in \cite{D1971}.
Both decomposable and uniform delay models are also applicable in the study of  road networks. The uniform model is applicable when roads are of the same or similar type, the delays having the same functional dependency on the flow vector.
As another application, in the context of internet traffic we propose the  use of this model when traffic is  prioritized  (applicable to the
{\em net-neutrality} debate \cite{O1999,P2007}).  
Users with higher priority are allowed to transmit more bits over a fixed period (unit) of time  or equivalently are allowed to transmit larger packets of information thus occupying the link for  a longer fraction of time. Consequently, the contribution of the delay induced by each type of user traffic is different.
% While the total delay faced on an edge is the same for each type of user, lower priority users who transmit less bits face higher delays per bit of transmission. 
We provide additional discussion of these motivating applications, including priority queue models, later in the section. A detailed analysis of priority models forms an interesting topic of further research.
% where the impact of favoring a particular type
% of flow, say $j$, can be modeled by increasing  its contribution to the
% delay via the co-efficient $a_j(e)$.

We study  Nash equilibrium and the price of anarchy in the defined heterogeneous delay model: our results show that the 
PoA worsens as compared with networks models where the delay is only a function of $f(e)$, the aggregate flow on the edge. This
indicates that type differentiation amongst traffic can lead to
a worse price of anarchy, a result that contrasts with the price of anarchy in the case of uniform and equal treatment of traffic types.

We consider both atomic and non-atomic flows, discussed above.
% In atomic flows, each flow request
% is required  to be satisfied on one path and in non-atomic networks the flow requests
% arise from multiple infinitesimal demands.
In both categories, we provide upper and lower bounds on the {\em price of anarchy} when there are $k$ types of
commodities in networks with heterogeneous and decomposable (including uniform) class of polynomial delay functions.
In particular, we show asymptotically tight bounds on the price of anarchy when the edges are  associated with affine
decomposable delay functions.

\paragraph{Previous Work:}
The study of equilibrium in
flow routing problems was
initiated by Pigou \cite{P1920},
and furthered by Wardrop \cite{W1952}.
Dafermos and Sparrow  \cite{DS1969} furthered this study and considered the relationship between flows that
are at Nash Equilibrium and flows that optimize a social
welfare function.
For  a  multi-commodity flow network
with convex delay functions on the edges, they establish a condition under which the flow satisfies
both Nash Equilibrium conditions and
also optimizes an aggregate objective function. In this model the delay is a function only of the total flow
through the edge. The assumption is that units traveling along a link uniformly share the
cost \cite{DS1969}. The interesting aspect of this work
is that for  homogeneous
polynomial costs that are a function of the aggregate variable $f(e)$, Nash equilibrium and social
optimum solutions coincide.
This is not so in a general setting.
Investigation of the inefficiency of Nash Equilibrium \cite{D1986} in network
flows led to an analysis
 of price of anarchy for load dependent delays \cite{KP1999}.
Roughgarden
and Tardos initiated the study of bounds on the price of anarchy \cite{RT2002}
in general networks followed by subsequent results of Roughgarden \cite{ROUGHGARDEN2003341,T2004,T2005} that use the Pigou bound. 
%first defined in that research. 
Additional results may be found by Correa et al. \cite{CSM2007,CSS2008} in their studies of
the price of anarchy with respect to the total latency as well as the average and the maximum latency. An aspect that distinguishes the research in \cite{CSS2008} is their use of geometric analysis in computing the price of anarchy.
The results in \cite{RT2002,ROUGHGARDEN2003341} show that for  polynomial delays, the price of anarchy is bounded by the degree of the polynomial and does not depend on the size of the network. This is indeed
a surprising result. As an extension, our results provide bounds for
the more general class of heterogeneous delay functions that depend on the characteristics of the delay function.
For the case when the {\em min-max delay} of paths is  used as the social welfare measure, results of  Weitz\cite{W2001} and in \cite{LRTW2011} show that for arbitrary increasing delay functions, the price of anarchy does
depend on the size of the network.

For atomic flows,
the price of anarchy  \cite{AAE2005}  has been  investigated for the typical unsplittable model with aggregate delay functions, i.e., delays  are independent of types of commodities. The results
show that the bound is dependent on the degree of the polynomial used to model the delay ($O(d^d)$ where $d$ is the degree of the polynomial).
Further, improved and exact bounds were proved on the worst-case PoA for unweighted and weighted in atomic unsplittable congestion games in  \cite{DT2006}. Techniques that depend on Pigou bounds 
%Roughgarden and Tardos 
do not carry over and the authors utilize the similarity of Nash equilibrium solutions with solutions obtained by a greedy algorithm for an on-line version of the problem. Note that Nash equilibrium solutions satisfy a variational inequality arising from local optimality conditions.
We utilize this approach to
provide  bounds on the price of anarchy for networks with generalized heterogeneous delay functions; the analysis is based  primarily  
on variational inequalities that arise in the atomic version of the problem.

The notion of commodity dependent delays was first studied
in Dafermos \cite{D1971,D1972} in which  the author considers a
transportation network where travel time on a link depends 
on the delays of the types of traffic.
One example of such a delay  function, for edge $e$ and commodity $i$, is provided by
$ \Phi^i_e(f_1,\ldots,f_k) = \sum_{j}g_{ij}(e)f_i(e)f_j(e) + h_i(e)f_i(e)$,
where $g_{ij}$ and $h_i$ are constants and $i,j$ represent commodities.
Dafermos \cite{D1971,D1972} establishes a condition ($[g_{ij}]$ must be a positive definite matrix) under which the flow (i) satisfies Nash Equilibrium conditions and (ii) optimizes an aggregate objective function.

Further related work is the research on player-specific congestion games \cite{M1996,ARV2006}, a model which allows for player specific
delays but which still uses the aggregate flow to compute delays.
% \todo{No need for next sentence}
% Other congestion games include malicious players \cite{BKP2007,BKP2009} and
% more recently, result on the price of anarchy
% have been illustrated when the delay function is defined by the $L_p$ norm \cite{CDT2012}

To our knowledge,  the price of anarchy
for  heterogeneous latency functions has not been studied before.

\paragraph{Motivation}
We outline motivations from two different areas, one being transportation engineering and the other
computer networks.

{ \em Transportation Networks}.
In transportation science, traffic equilibrium has been considered extensively to analyze transportation systems,
starting with the work of Wardrop.
Note that this prior work has limitations in modeling multiple vehicle types or multi-modal networks.
In order to provide more accurate modeling, transportation science models were extended  to the multi-class
model where the demand is partitioned into multiple classes and the cost of a route is dependent on the traffic
of all user classes along the route\cite{D1972}.

Multi-class models have been defined to extend the notion of Wardrop equilibrium, by Dafermos and Potts and Oliver\cite{PO1972,D1972}, where in
the latter paper they used the generalized BPR (Bureau of Public Roads) model that depend on multiple class parameters.
Multi-class models may also be found in \cite{N2000}.

{\em Computer Networks}.
Another 
%One of the primary 
motivation of this research comes from a study of the impact of providing
differentiated service to different types of traffic. Suppose we have $k$ types of traffic and traffic types
are given priorities.
The delay function for the class with the $k$th priority in a $M/G/1$
queuing system is
$E(T_k)=E(R)/(1- \sum_{1\leq j \leq k-1} \rho_j)(1-\sum_{1 \leq j \leq k} \rho_j)$,
where $E(R)$ is the mean residual service time and $\rho_k$ is the load of the class-$k$ traffic.
The average delay  is obtained from the expected delay of all commodities
 \cite{A1954}.
As a first attempt to study selfish behavior in queues with differentiated service rates, these delays may be modeled by a polynomial, possibly a monomial using the dominant term obtained from a Taylor series expansion.
The  study of polynomial delay functions that is reported
in this paper is applicable in this context.

Furthermore, in this paper we utilize the traditional routing
model where users are able to plan routing paths based on the delays in the network.
These results thus apply to transportation network and computer network routings that rely on the link state of the network.
A possible extension would be to find the price of anarchy when other  internet routing models are used  %by Papadimitriou and Valiant 
\cite{PV2010}.
% \todo{No need for this sentence. What is this}

Before we outline our results we present our model and some formal definitions:

\subsection{Model and Definitions}
\label{ModelandDef}
We consider a directed network $G = (V,E)$ with a vertex set $V$, an edge set $E$, and a set $K$ of $k \in \mathbb{N}$ source-destination pairs $\{s_1, t_1\}, \ldots, \{s_k, t_k\}$, each corresponding to a different traffic type.
We do not allow self-loops for any vertex but  accept parallel edges between any pair of vertices.

For commodity $i$, denote the set of (simple) $s_i$-$t_i$ paths by $\mathcal{P}^i$, and define a set of paths $\mathcal{P} = \bigcup\mathcal{P}^i$.
For any path $P\in \mathcal{P}^i$,  define a flow,
$f^i_P \in \mathbb{R}^+$ to be the flow of commodity $i$ that is assigned to path $P$.
The set of all flows is represented by the flow  vector
$f = (f^i_P)_{i \in K, P \in \mathcal{P}^i}$.
For a fixed flow $f$, commodity $i$ and an edge $e \in E$  denote the flow of
commodity $i$ through edge $e$
by $f_i(e) = \sum_{P \in \mathcal{P}^i | e \in P}f^i_P$ and the flow  vector on the
edge by $f(e) = (f_i(e))_{i \in K}$. The total aggregate flow through edge $e$
is denoted by $||f(e)||_1 = \sum_if_i(e)$.
The demand requirement of commodities $1$ through $k$ is denoted by $R = (r_1, \ldots, r_k)$ where $r_i \in \mathbb{N}$ represents the amount of flow required to be routed from source $s_i$ to destination $t_i$.
A {\em feasible} flow $f$ is one where the demand requirement for each commodity
is satisfied, i.e., $\sum_{P \in \mathcal{P}^i}f^i_P = r_i, \forall i$.

We  let $\Phi_e : \mathbb{R}^{k} \rightarrow \mathbb{R}^{+}$ be a
heterogeneous delay function, as defined below, on edge $e$, and assume that $\Phi_e(f(e)), f(e)\in \mathbb{R}^k$ is  monotonically increasing w.r.t  components of $f(e)$ % \textcolor{red}{w.r.t $||X||_1$} \todo{what??}, 
and convex.
The delay experienced by commodity $i$ on a path $P \in \mathcal{P}^i$ is defined to be
\[\Phi_{P}^i(f) = \sum_{e \in P} \Phi^i_e(f_1(e),\ldots,f_k(e))\]
%= \sum_{j}g_{ij}(e)f_j(e) + c_i(e)\]
where $\Phi^i_e(f_1(e),\ldots,f_k(e))$ is the delay faced by commodity $i$ on edge $e$, given the flow vector $f$ and will
be represented by $\Phi^i_e (f(e))$.
%and $g_{ij}$ and $c_i$ are constants and $i,j$ represent commodities. We will simplify notation and refer to $(f_1(e),\ldots,f_k(e))$ by $f(e)$ also.
Furthermore, $f_i(e)$ will be referred to as $f_i$ when $e$ is evident from the context.
% * <kapoor@iit.edu> 2016-07-23T16:39:39.892Z:
%
% ^.

A heterogeneous delay function is a generalization of the standard  aggregate delay function. %$\widetilde{\Phi}$.
We consider polynomial {\em heterogeneous} delay
functions where the delay function on edge $e$ for commodity $i$ is
represented by
%$\Phi^i_e(f_1(e) , f_2(e), \ldots, f_k(e)) + c_i(e)$
$\Phi^i_e(f(e))= \sum_{1 \leq \ell \leq N^i_e } g_{i\ell}(e)h_{i\ell}(f,e) + c_i(e)$
%where $f_j(e)$ is the flow of commodity $j$ on edge $e$,
where $\forall (i, \ell, e), \ g_{i\ell}(e)$ and $c_i(e)$ are nonnegative real valued constants.
%and $i$ represents a  commodity. 
Here, $h_{i\ell}(f,e)$ %%f_1, f_2, \ldots, f_k)$ 
is any arbitrary monomial (with co-efficient one) in terms of the flow vector $f$ and dependent on edge $e$; and there are $N^i_e$ terms
in the polynomial $\Phi^i_e(f)$.
% , since uniformly
% scaling the coefficients and constants preserves the relative cost of the paths in the flow pattern.
%and nonnegative $c_i(e)$ are constants and
 % of degree $\theta$.
Unfortunately, these functions will be shown to have unbounded price of anarchy (see \Cref{sec:unbounded-poa}, in particular \Cref{lem:unboundedPoA}).
%Lemma \textcolor{red}{AddLower bound references forward}).

As in \cite{AAE2005}, we study polynomials  that 
% all terms have degree $\theta$
% and that contain independent terms for each variable, i.e. 
are  of the form 
$\sum_i a_i x_i^{\theta}$, where $\theta \in \mathbb{N}$ and $\forall i : a_i \in \mathbb{R}^{++}$. 
%\todo{Does $a_i$ have to be positive} 
%\todo{J: isn' this be $\mathbb{R}^+$ for consistency with others} 
We actually consider a class of more general polynomials that extend this form. 
Any polynomial delay that does not have the defined structure of terms will be shown to have a price of anarchy that is dependent on requirements or is unbounded.
We formalize this set of polynomials below.
 
We define a delay to be a $\theta$-complete polynomial delay (or simply  $\theta$-polynomial delay), $\theta \in \mathbb{N}$, if the
delay function is of the form:
$$\Phi^i_e(f) = \sum_{ 1 \leq j\leq k} a_{ij}(e) f_j^\theta+ %\sum_i
\sum_{ 1 \leq \ell \leq L^i_e}g_{i\ell}(e)f^{\theta_1^\ell}_{1}f^{\theta_2^\ell}_{2}\cdots f^{\theta_k^\ell}_{k} + c_i(e)$$ 
where $L^i_e \in \mathbb{N}$ and 
$ \forall \ell , \theta_j^\ell \in \mathbb{N},  \ \theta_j^\ell < \theta$ with $ \sum_{j=1}^{k} \theta_j^\ell  = \theta $. 
%\textcolor{red}{$\leq \theta$} 
Also, $\forall (i,j,e)$: $ a_{ij}(e) \in \mathbb{ R}^{++}$ and   %\todo{J: isn' this be $\mathbb{R}^+$ for consistency with others} 
$ \forall (i,\ell,e)$: $ \ g_{i\ell}(e), c(e) \in \mathbb{R}^+$. %_{\geq 0}$,
Let $N^i_e=L^i_e+k+1$  denote the number of terms in the polynomial $\Phi^i_e$.
% and which, by the assumptions of $\theta$-complete polynomials,
% includes the terms $\sum_i a_i f_i^\theta$.
% We term such functions as
% The delays will be expressed as:
%\todo{Note that $N_{\Phi}$ includes all terms now--so later on some lemma can be simplified}
Further, we use $N_{\Phi}= \max_{i,e} N^i_e$ to represent the maximum number of terms in the polynomial $\Phi^i_e(f)$ over all $i \in K$ and $e \in E$. Note that for linear functions, $N_{\Phi}$ is bounded by $k+1$, the number of traffic types.
We assume w.l.o.g that all positive coefficients and constants are scaled uniformly to be greater than one.
%i.e., $g_{il}(e) \geq 1$ and also $c_i(e) \geq 1$. 
% For convenience of notation, throughout this paper, we assume that the minimum co-efficient of the polynomials is 1, and 
Moreover, with this assumption, the maximum value of the ratio of the co-efficients and constants,  i.e. $ \max \{ \max_{i,\ell,e} g_{i\ell}(e), \max_{i,j,e} \{ a_{ij}(e) \}, \max_{i,e} c_{i}(e) \}$, is denoted by $a_{max}$.

%Moreover, the maximum value of the co-efficients and constants,  $ \max \{ \max_{i,l,e} g_{il}(e), \max_{i,j,e} \{ a_{ij}(e) \}, \max_{i,e} c_{i}(e) \}$, is denoted by $a_{max}$.
% we let $a_{max}$ denote the fraction 
% $  \frac{\max _{i,j,e} \{a_{ij}(e) , g_{ij}(e) \}}{ \min_{i,j,e} \{ a_{ij}(e), g_{ij}(e) \}}$}.

% $  \frac{\max \{ \max_{ij} a_{ij}(e), \max _{i,j,e}g_{ij}(e) \}}{\min \{ \min_{ij} a_{ij}(e), \{ \min_{i,j,e}g_{ij}(e) \}}$.
% where $g_{ij}(e)$ is the coefficient of the $j$th term in the polynomial $h^e_i()$.

% The total delay on edge $e$ that commodity $i$ incurs is as follows:
% \[ f_i(e)\Phi^i_e(f_1,\ldots,f_k) = f_i(e)( h^e_i (f_1, \ldots, f_k) + c_i(e)) \]

In the case of heterogeneous delay functions, as defined above, each commodity has a  different,
commodity related, delay function on the same edge $e$.
However, under the condition that $\forall (i,e)$: $g_{i\ell}(e) = g_{\ell}(e), \ h_{i\ell}(e) = h_{\ell}(e)  $ and $\forall (i,j,e)$:  $c_i(e) =c_j(e)=c(e), \ $%\Phi^i_e=\Phi^j_e=\Phi_e , \ $ 
the polynomial heterogeneous delay function $\Phi^i_e(f)$ (with $L_e \in \mathbb{N}$ monomial terms) is expressible as
\[ \Phi_e(f_1, \ldots, f_k) = \sum_{1\leq \ell \leq L_e}g_{\ell}(e)h_{\ell}(f,e) + c(e) \]
% where $\forall l$ and $\forall e,  \  g_{l}(e)$ and $c(e)$ are positive constants. 
Such a delay function is termed as {\em decomposable}.

Additionally, we define {\em uniform} functions as decomposable delay functions 
of the form: $\Phi_e(f)= \sum_{1 \leq \ell \leq L_{\Phi}} g_{\ell}h_{\ell}(f)+ c(e)$,
where the delay function (with $L_{\Phi}$ monomial terms) is the same function  of the flow, across all edges.
%independent of  particular edge.
Note that the constant term is still edge dependent.
% only  dependent on the flow vector on the edge and is 
%f_1, f_2, \ldots, f_k) 
Such a function may be used to model a network where the queues on each edge have the same parameters arising from using  the same  underlying technology.
%and thus the delay functions are not edge dependent.

Both, non-atomic and atomic flows are considered in this paper. In non-atomic flows, the flow of the $i$th commodity on a path $P \in {\cal P}_i$ can be arbitrarily small,
representing the fact that the player corresponding  to any one commodity controls a negligible amount of traffic, resulting in multiple paths being used to satisfy requirements. In atomic flows, each commodity $i$ can use only one path to push the required units of flow $ r_i \in \mathbb{N}$.

We represent a network flow routing game by the triple  ${\cal G} =(G,R,\Phi)$. For a game ${\cal G}$, 
we define the social cost to be the sum total of all delays on edges over  flows
of all commodities, i.e.,
$C(f) = \sum_{e} \sum_i f_i(e)\Phi^i_e(f(e))$.
%f_1(e), f_2(e), \ldots, f_k(e))$.
Let $f_{NE}$ and $\hat{f}$ be a Nash Equilibrium and a social optimum flow, respectively and
%Then the worst case Nash Equilibrium cost is denoted by  $C_{NE}(f) = \sup_{f \in {\cal E}}C(f)$, where ${\cal E}$ is the set of
%all Nash Equilibrium solutions
let the social optimum cost $C_{SO}(\hat{f}) = \min_{f}C(f)$ where $\hat{f} = \arg\min_{f}C(f)$.
The price of anarchy is defined as:
\[ PoA({\cal G}) =  \sup_{f_{NE} \in {\cal N}({\cal G})}\frac{C(f_{NE})}{C_{SO}(\hat{f})}\]
where ${\cal N}({\cal G})$ is the set of
all Nash Equilibrium solutions of the game ${\cal G}$.
% With respect to an instance $(G, R, \Phi)$, a social optimum flow
% % feasible flow minimizing $C(f)$ is said to be an optimal or minimum-delay flow; such a flow 
% always exists because the space of all feasible flows is a compact set and our cost function is continuous.

\subsection{Results}
Our study of equilibrium flows in networks with heterogeneous delay functions provides multiple results on the existence of equilibrium and the price of anarchy. We describe below, the results that we obtained for both the  the non-atomic and atomic models.

\paragraph{Results for Non-atomic Network Flows:}
To analyze this class of problems, we introduce a multi-dimensional version of the Pigou bound \cite{RT2002} that includes  a 
parameter which we term as the {\em Gamma} bound.
Combining the two, provides us bounds on the price of anarchy for non-atomic flows.

We  show that the price of anarchy for the general class of convex heterogeneous functions
can become unbounded, even growing with the flow requirement. This contrasts with results for delay functions that only depend on  the aggregate flow, a model for which Roughgarden and Tardos \cite{RT2002} bound the price of anarchy in terms of a parameter $\alpha$ that measures the growth rate of the  delay function.

However for the class of heterogeneous $\theta$-complete polynomial delay functions, which include the standard affine functions,  we provide bounds on the price of anarchy that are independent of the network size and the flow requirements. 
%We define
% $\theta$-polynomial delay functions to be multivariate polynomial functions where every term has degree $\theta$ and
% which contain independent terms for each commodity (absence of independent terms leads to unbounded price of anarchy). This includes the class of linear functions.
 Our results show that the price of anarchy is bounded by a function
of the polynomial size, degree and coefficients. We also discuss lower bound examples.
These results extend previous work on the price of anarchy, more specifically the results
that \cite{RT2002} presents. 
The existence of Nash Equilibrium for convex heterogeneous delays follows from  the results in \cite{RT2002} which
show  the  existence of a Nash Equilibrium flow in $k$-commodity non-atomic networks
with convex  delay functions.

In summary (also refer to \Cref{table:k-nonatomic}):
%\todo{add lower bound}

\begin{itemize}

\item
% To get sharper bounds for  decomposable $\theta$-polynomial functions 
We generalize the Pigou bound \cite{ROUGHGARDEN2003341,RT2002} and provide a multi-dimensional version,  which is termed the {\em Gamma-Pigou } bound (section~\ref{sec:Gamma-Pigou bound}). This bound utilizes an independent term, defined as the {\em Gamma} bound. We show that the price of anarchy of selfish non-atomic routing in networks with heterogeneous delay functions can be obtained from combining these bounds.
% Combining the two provides us tight bounds.

\item
We provide examples (section~\ref{sec:unbounded-poa}) to illustrate that the price of anarchy  for heterogeneous functions
can become unbounded in the general case.

\item We provide a bound on the price of anarchy in networks with $k$ (types of) commodities where each edge
is associated with heterogeneous $\theta$-complete polynomial delay functions.
This bound is $(a_{max}kN_{\Phi})^{\theta+1}$ in the most general form of the network where
$N_{\Phi}$ represents the maximum number of terms in the delay function over all edges (section~\ref{subsec-poa-hetero-poly}). We note that this bound is independent of the network size but grows with the number of
commodity types and exponentially with respect to $\theta$, the degree of the polynomial. A similar behavior is exhibited in
%For notational convenience, throughout this paper, we let $a_{max} = \frac{\max_{i,j,e}g_{ij}(e)}{\min_{i,j,e}g_{ij}(e)}$
%where $g_{ij}(e)$ is the coefficient of the $j$th term in the polynomial $h^e_i()$.
the case of decomposable $\theta$-complete polynomial delay functions (section~\ref{sec-nonatmoic-poa}), where the price of anarchy bound is
$  (\theta+1)a_{max}k^{\theta-1}$, which becomes $2a_{max}$ for affine functions.

When we consider networks with heterogeneous affine (i.e., $\theta = 1$) delay functions, the bound is 
$  k^2N^2_{\Phi}a^2_{max}+2$. %$k^4a^2_{max}+2$
Note that $N_{\Phi}$ is bounded by $k+1$. 
These bounds should be contrasted with the bound of $4/3$ achieved when the delay on the edge is linearly
dependent on the aggregate flow of the commodities on the edge.

The bounds illustrate an exponential growth with the degree of the polynomials and linear or polynomial growth with the value of $a_{max}$.
% $\frac{k^2N^2_{\Phi}a^2_{max}+\sqrt{(k^2N^2_{\Phi}a^2_{max})(k^2N^2_{\Phi}a^2_{max}+4)}}{2} + 1$ where $N_{\Phi}$ represents maximum number of terms in the polynomial representing $\Phi_e()$ over all $e$.
%given delay function $\Phi()$.
% \item
% To get sharper bounds for  decomposable $\theta$-polynomial functions we introduce a multi-dimensional version of the
% Pigou bound \cite{RT2002} as  well as another bound which we term as the {\em Gamma} bound.
% Combining the two provides us tight bounds.
\item
We provide an example(lemma~\ref{lem:LowerBoundPOA}) that illustrates the tightness of our bounds. The  example network has a POA of  $a_{max}(k-1)^{\theta-1}$ for large enough value of $a_{max}$. This network provides an asymptotically tight bound for affine functions.
\end{itemize}
\begin{table}[ht]
\begin{center}
\scalebox{1.0}{
\begin{tabular}[width=\textwidth]{|c|cc|}
\hline
 & \multicolumn{2}{|c|}{$k$ commodities} \\ [-1ex]
\raisebox{2.0ex}{PoA} & affine & $\theta$-complete polynomial \\ [1ex]
\hline
\hline
Aggregate~ & $4/3$~\cite{RT2002}& $\Theta\left(\frac{\theta}{\log \theta}\right)~$\cite{ROUGHGARDEN2003341}  \\
\hline
 & \multicolumn{1}{|c}{$2a_{max}$} & \multicolumn{1}{c|}{$(\theta+1)a_{max}k^{\theta-1}$}
 \\ [-1ex]
\raisebox{2.0ex}{Decomposable} & [\Cref{thm-2-nonatomic-decomposable-poly}*] & [\Cref{thm-2-nonatomic-decomposable-poly}] \\ [1ex]
\hline 
 & \multicolumn{1}{|c}%{$(k^2N^2_{\Phi}a^2_{max}+\sqrt{(k^2N^2_{\Phi}a^2_{max})(k^2N^2_{\Phi}a^2_{max}+4)})/2 + 1$} 
{$k^2N^2_{\Phi}a^2_{max}+2$} & \multicolumn{1}{c|}{$(a_{max}kN_{\Phi})^{\theta+1}$}  \\ [-1ex]
\raisebox{2.0ex}{Heterogeneous} & [\Cref{thm-k-nonatomic-heterogeneous-affine}] &  [\Cref{thm-k-nonatomic-heterogeneous-poly}]  \\ [1ex]
\hline
\end{tabular}}
\center{Lower Bound [Lemma~\ref{lem:LowerBoundPOA}]: $\Omega(a_{max}(k-1)^{\theta-1})$}
\caption{PoA for Nonatomic Flows (* indicates an asymptotically tight bound)}

\label{table:k-nonatomic}
\end{center}
\end{table}

\paragraph{Results for Atomic Network Flows:}
Our second set of results consider  atomic flows. In the most general setting Nash equilibrium does not exist and thus wherever necessary we will assume that the instance provided has a Nash equilibrium.
Our results for the case of atomic flows are summarized below (also refer to \Cref{table:k}).
\begin{itemize}
\item
We show the existence of Nash equilibrium for uniform heterogeneous affine delay functions. Our proof uses a potential function argument.
Further, we show that for non-uniform heterogeneous  functions, where edges have
different delay functions, Nash Equilibrium  need not  exist (section~\ref{subsec:ExistenceAtomic}). The second result  is indeed surprising, especially
since it only requires affine functions and unit demands (unweighted) to generate the examples.
This contrasts with the fact that any unweighted congestion game with homogeneous
delays, where the delay is a function of aggregate flows, has a pure Nash equilibrium \cite{R1973, MS1996}

\item We consider the price of anarchy for atomic flow routing in networks with $k$ (types of) commodities where each edge
is associated with decomposable $\theta$-complete polynomial delay functions (section~\ref{subsec-poa-general-atomic-poly}).

When we consider networks with decomposable affine delay
functions this bound is $a_{max}+2$ and 
thus is dependent on the relative impact of the commodities.
This contrasts with the constant PoA bound of $\frac{3+\sqrt{5}}{2}$ ~\cite{AAE2005} when the delay function depends only on the aggregate flow.

% $\frac{a_{max}}{2}+\frac{(a_{max}^2+4a_{max})^{1/2}}{2}+1$.
In the case of decomposable $\theta$-complete polynomial delay functions, the bound is $O(a_{max}^{\theta+2}N_{\Phi}^{\theta+1})$. The above bounds utilize arguments based on a variational inequality that arises in these problems and satisfied by any Nash equilibrium solution. %\todo{Changed} 

Again, even in
this case, the PoA does not depend on the network size but only on the characteristic of the delay function. However, the relative co-efficients in the delay function, $a_{max}$, is a substantial factor.
Contrast this with the bound when the delays are a function of the aggregate flows: the price of anarchy for aggregate delay functions is $\Theta((\frac{\theta}{\log \theta})^{\theta+1}) $~\cite{DT2006}, independent of the relative weights of the types of  commodities in the network or the number of commodities themselves.

\end{itemize}
\begin{table}[ht]
\begin{center}
\scalebox{1.0}{
\begin{tabular}[width=\textwidth]{|c|cc|}
\hline
\mbox{PoA} & affine & $\theta$-complete polynomial \\
\hline\hline
Aggregate  & $\frac{3+\sqrt{5}}{2}$~\cite{AAE2005} & $\Theta\left(\frac{\theta}{\log \theta}\right)^{\theta+1}$~\cite{DT2006} \\
\hline
 & \multicolumn{1}{|c}{$a_{max}+2$ }%\sqrt{a_{max}^2+4a_{max}}}{2} + 1^{*}$}
& \multicolumn{1}{c|}{$O(a_{max}^{\theta+2}N_{\Phi}^{\theta+1})$} \\ [-1ex]
\raisebox{2.0ex}{Decomposable} & [\Cref{thm-k-atomic-decomposable-affine}*] & [\Cref{thm-k-atomic-decomposable-poly}] \\ [1ex]
\hline
\end{tabular}}
\center{Lower Bound [Lemma~\ref{lem:LBAtomic}]: $\Omega(a_{max})$}
\caption{PoA for $k$-commodity Atomic Flows (* indicates an asymptotically tight bound)}
\label{table:k}
\end{center}
\end{table}
%\todo{Where is * used in Table}
% For $2$-commodity networks with uniform affine delay functions, we show (\Cref{table:2}) improved bounds which are different from the bounds on general decomposable affine delay functions.
Further, for atomic flows in networks with two different types of commodities
and uniform delay functions, i.e., delays defined by $a_1f_1(e) + a_2f_2(e) + c(e)$,
we show an improved bound on the  price of anarchy (\Cref{table:2}). This is achieved by
considering the cycles that arise when the Nash equilibrium and social optimum flows
are overlaid, canceling common flows.
By analyzing the flow in these cycles, we provide an asymptotically tight bound of $\sqrt{a_{max}} + 2$.
\begin{table}[ht]
\begin{center}
\scalebox{1.0}{
\begin{tabular}[width=\textwidth]{|c|cc|}
\hline
\mbox{PoA} & uniform & decomposable \\
\hline\hline
 & \multicolumn{1}{|c}{$\sqrt{a_{max}} + 2$} & \multicolumn{1}{c|}{$a_{max}+2$}%\sqrt{a_{max}^2+4a_{max}}}{2} + 1$} 
 \\ [-1ex]
\raisebox{2.0ex}{Affine} & [\Cref{thm-2-atomic-uniform-affine}*] & [\Cref{thm-k-atomic-decomposable-affine}] \\ [1ex]
\hline
\end{tabular}}
\center{Lower Bound [Lemma~\ref{lower-bound-affine-uniform-atomic}]: $\Omega(\sqrt{a_{max}})$}
\caption{PoA for $2$-commodity Atomic Flows}(* indicates an asymptotically tight bound) \label{table:2}
\end{center}
\end{table}

% We present preliminaries in Section \ref{prelim} and discuss relevant properties.
\subsection{Organization of the paper}
In \Cref{section:nonatomic}, we consider bounds for the price of anarchy for non-atomic networks with decomposable delay functions.
Further, in \Cref{secAtomic} we present our results for atomic flows  and strengthen our results on the  price of anarchy for two-commodity network flows with uniform delay functions in \Cref{sec-poa-2com-atomic}.

\section{Part I: The Non-Atomic Case}\label{section:nonatomic}
In this section we consider the case of non-atomic network flow routing games. We consider heterogeneous functions
with polynomial delays.

Our contributions %on the upper bound
on bounds for the price of anarchy in non-atomic selfish routing problems have been summarized  in  \Cref{table:k-nonatomic}.

For general heterogeneous functions, the price of anarchy is unbounded as will be illustrated in \Cref{sec:unbounded-poa}.
As the table shows, we obtain upper bounds on the price of anarchy for heterogeneous delays for the
interesting  cases where  the function is decomposable and when we consider $\theta$-complete polynomial delays.

To show  bounds on PoA we will use the following {\em variational inequality characterization} \cite{D1980},
the proof of which is simple and omitted.
\begin{lemma}
\label{lem:variational}
Let $f$ be a feasible flow for the non-atomic instance $(G,R,\Phi)$. The flow $f$ is a Nash equilibrium flow if and only if
\[ \sum_e\sum_if_i(e)\Phi^i_e(f_1,\ldots,f_k) \leq \sum_e\sum_i\hat{f}_i(e)\Phi^i_e(f_1,\ldots,f_k)\]
for every flow $\hat{f}$ feasible for $(G,R,\Phi)$.
\end{lemma}

\subsection{The Gamma-Pigou Bound for the Price of Anarchy}
\label{sec:Gamma-Pigou bound}
%\todo{this section applies in general--lets consider to move it up}
We first establish an approach to determining the upper bounds on the price of anarchy for multidimensional delay functions.
For simplicity we consider decomposable functions.
While this theory is unrestricted, in that the bounds apply for arbitrary heterogeneous functions also, analysis indicates that the price of anarchy is unbounded in these cases.
To simplify notations, we use $f$ instead of $f(e)$ when clear from the context.

In the analysis of  heterogeneous delay functions the standard Pigou bound arguments, designed for aggregate delay functions, do not apply.
We thus consider a multi-dimensional version of the Pigou bound. Since that bound, by itself, does not suffice,  we establish that the price of anarchy is bounded by  the generalized Pigou bound and the {\em Gamma} bound, which we describe
in detail below.
In the subsequent subsection we analyze these bounds in the context of the delay functions
defined by decomposable $\theta$-complete polynomial functions.

We first describe the multi-dimensional version of the Pigou bound \cite{RT2002} as applied to flow vectors that represent the $k$ types of commodities.
\begin{definition}
(Generalized Pigou bound) Let $\mathcal{C}$ be a nonempty set of cost functions, each function defined over $\mathbb{R}_+^k$.
The Pigou bound $\alpha(\mathcal{C})$ for $\mathcal{C}$ is
\[ \alpha(\mathcal{C}) = \sup_{\Phi \in \mathcal{C}}\sup_{x,r \geq 0}\frac{\gamma(\mathcal{C}) ||r||_{1}\Phi(r_1,\ldots,r_k)}{\beta(\mathcal{C})}, \]
\[ \beta(\mathcal{C}) =\gamma(\mathcal{C})||x||_{1}\Phi(x_1,\ldots,x_k)+(||r||_{1}-||x||_{1})\Phi(r_1,\ldots,r_k),\] where $r=(r_1,\ldots,r_k)$ and $x=(x_1,\ldots,x_k)$, $r,x \in \mathbb{R}_+^k$ with the understanding that $0/0 = 1$,
%$r_1,\ldots,r_k$ and $x_1,\ldots,x_k$ are nonnegative
and
\[ \gamma(\mathcal{C}) = \sup_{\Phi \in \mathcal{C}}\sup_{r,x: ||x||_1 \geq ||r||_1} \frac{\Phi(r_1,\ldots,r_k)}{\Phi(x_1,\ldots,x_k)}
\]
\end{definition}
We split our analysis into two cases: 1)$||x||_{1} < ||r||_{1}$ and 2)$||x||_{1} \geq ||r||_{1}$.
In the case when $||x||_{1} < ||r||_{1}$, we consider the generalized Pigou bound as described
above.
However, we will evaluate the Pigou bound in a restricted setting by using the
following mathematical program (SC):
$$
\mbox{\bf SC: } \alpha(\mathcal{C}) = \sup_{\Phi \in \mathcal{C}}\sup_{x,r \geq 0}\frac{\gamma(\mathcal{C})||r||_{1}\Phi(r_1,,\ldots,r_k)}{\gamma(\mathcal{C})||x||_{1}\Phi(x_1,\ldots,x_k)+(||r||_{1}-||x||_{1})\Phi(r_1,\ldots,r_k)
} 
$$
\begin{eqnarray*}
\mbox{ subject to}\phantom{a} \Phi(x_1,\ldots,x_k) &\leq& \Phi(r_1,\ldots,r_k)\\
\phantom{a} ||x||_{1} & < & ||r||_{1}
\end{eqnarray*}
Note that when $\Phi(x_1,\ldots,x_k) > \Phi(r_1,\ldots,r_k)$, the value of  $||x||_{1}(\Phi(x_1,\ldots,x_k) - \Phi(r_1,\ldots,r_k))$ becomes positive and decreases the value of $\alpha(\mathcal{C})$.
Thus w.l.o.g we add the constraint  that $\Phi(x_1,\ldots,x_k) \leq \Phi(r_1,\ldots,r_k)$.%\todo{ justify via monotonicity}

For the second case, we will consider the second ratio, termed the {\em Gamma Bound}, defined above and restated in a form that will be used in the proof below:
\begin{eqnarray*}
\phantom{a} \gamma(\mathcal{C})
\phantom{a} &=& \sup_{\Phi \in \mathcal{C}}\sup_{r,x: ||x||_1 \geq ||r||_1} \frac{\Phi(r_1,\ldots,r_k)||r||_{1}+\Phi(r_1,\ldots,r_k)(||x||_{1}-||r||_{1})}{\Phi(x_1,\ldots,x_k)||x||_{1}}
\end{eqnarray*}
% In our proof below, $\hat{f}$  and $f$  will indicate optimal and Nash equilibrium flows on edge $e \in E$.
Given flow vectors $f$ and $\hat f$, we will let $E_{||f(e)|| > ||\hat{f}(e)||}$ be the set of all edges with the property that $||f(e)||_{1} > ||\hat{f}(e)||_{1}$.

We are now able to express the price of anarchy in
terms of the above {\em Pigou} and {\em Gamma} bounds:
\begin{theorem}[Gamma-Pigou Bound]\label{thm:poa-commodity-affine}
Let $\mathcal{C}$ be a set of %heterogeneous decomposable polynomials and $\alpha(\mathcal{C}), \gamma(\mathcal{C})$ the Pigou, Gamma bounds, respectively, for $\mathcal{C}$.
decomposable monotonically increasing delay functions. 
If $(G,R,\Phi)$ is an instance of
a non-atomic k-commodity
network flow routing game
with delay function $\Phi \in \mathcal{C}$, then the price of anarchy in $(G,R,\Phi)$
is at most $\max\{\alpha(\mathcal{C}),\gamma(\mathcal{C})\}$.
\end{theorem}
\begin{proof}
Let $\hat{f}$ and $f$ be the optimal and Nash equilibrium flow, respectively,
for a non-atomic instance $(G,R,\Phi)$ with delay functions in the set $\mathcal{C}$.
We let the total cost of flow $f$ be $C_{NE}(f)$.
To prove the theorem we note that:
\begin{eqnarray*}
% \phantom{a} C_{SO}(\hat{f}) &=& \sum_{e \in E}\Phi_e(\hat{f}_1(e),\ldots,\hat{f}_k(e))||\hat{f}(e)||_{1}\\
\phantom{a} C_{SO}(\hat{f}) &=& \sum_{e \in E}\Phi_e(\hat{f}(e))||\hat{f}(e)||_{1}\\
% \phantom{a} &\geq& \frac{1}{\alpha(\mathcal{C})}\sum_{e \in E_{||f|| > ||\hat{f}||}}\Phi_e(f_1(e),\ldots,f_k(e))||f(e)||_{1} +  \\
\phantom{a} &\geq& \frac{1}{\alpha(\mathcal{C})}\sum_{e \in E_{||f|| > ||\hat{f}||}}\Phi_e(f(e))||f(e)||_{1} +  \\
% \phantom{a} & & \frac{1}{\gamma(\mathcal{C})}\sum_{e \in E_{||f|| > ||\hat{f}||}}\Phi_e(f_1(e),\ldots,f_k(e))(||\hat{f}(e)||_{1} - ||f(e)||_{1}) + \\
% \phantom{a} & & \frac{1}{\gamma(\mathcal{C})}\sum_{e \in E_{||f|| \leq ||\hat{f}||}}\Phi_e(f_1(e),\ldots,f_k(e))||f(e)||_{1} + \\
% \phantom{a} & & \frac{1}{\gamma(\mathcal{C})}\sum_{e \in E_{||f|| \leq ||\hat{f}||}}\Phi_e(f_1(e),\ldots,f_k(e))(||\hat{f}(e)||_{1}-||f(e)||_{1})\\
% \phantom{a} &\geq& \frac{1}{\alpha(\mathcal{C})}\sum_{e \in E_{||f|| > ||\hat{f}||}}\Phi_e(f_1(e),\ldots,f_k(e))||f(e)||_{1} + \\
% \phantom{a} & & \frac{1}{\gamma(\mathcal{C})}\sum_{e \in E_{||f|| \leq ||\hat{f}||}}\Phi_e(f_1(e),\ldots,f_k(e))||f(e)||_{1} + \\
% \phantom{a} & & \frac{1}{\gamma(\mathcal{C})}\sum_{e \in E_{||f|| > ||\hat{f}||}}\Phi_e(f_1(e),\ldots,f_k(e))(||\hat{f}(e)||_{1} - ||f(e)||_{1}) + \\
% \phantom{a} & & \frac{1}{\gamma(\mathcal{C})}\sum_{e \in E_{||f|| \leq ||\hat{f}||}}\Phi_e(f_1(e),\ldots,f_k(e))(||\hat{f}(e)||_{1}-||f(e)||_{1}) \\
% \phantom{a} &\geq& \frac{1}{\alpha(\mathcal{C})}\sum_{e \in E_{||f|| > ||\hat{f}||}}\Phi_e(f_1(e),\ldots,f_k(e))||f(e)||_{1} + \\
% \phantom{a} & & \frac{1}{\gamma(\mathcal{C})}\sum_{e \in E_{||f|| \leq ||\hat{f}||}}\Phi_e(f_1(e),\ldots,f_k(e))||f(e)||_{1} \\
% \phantom{a} &\geq& \min\{\frac{1}{\alpha(\mathcal{C})},\frac{1}{\gamma(\mathcal{C})}\}\sum_{e \in E}\Phi_e(f_1(e),\ldots,f_k(e))||f(e)||_{1} = \frac{C_{NE}(f)}{\max\{\alpha(\mathcal{C}),\gamma(\mathcal{C})\}}.
\phantom{a} & & \frac{1}{\gamma(\mathcal{C})}\sum_{e \in E_{||f|| > ||\hat{f}||}}\Phi_e(f(e))(||\hat{f}(e)||_{1} - ||f(e)||_{1}) + \\
\phantom{a} & & \frac{1}{\gamma(\mathcal{C})}\sum_{e \in E_{||f|| \leq ||\hat{f}||}}\Phi_e(f(e))||f(e)||_{1} + \\
\phantom{a} & & \frac{1}{\gamma(\mathcal{C})}\sum_{e \in E_{||f|| \leq ||\hat{f}||}}\Phi_e(f(e))(||\hat{f}(e)||_{1}-||f(e)||_{1})\\
\phantom{a} &\geq& \frac{1}{\alpha(\mathcal{C})}\sum_{e \in E_{||f|| > ||\hat{f}||}}\Phi_e(f(e))||f(e)||_{1} + \\
\phantom{a} & & \frac{1}{\gamma(\mathcal{C})}\sum_{e \in E_{||f|| \leq ||\hat{f}||}}\Phi_e(f(e))||f(e)||_{1} + \\
\phantom{a} & & \frac{1}{\gamma(\mathcal{C})}\sum_{e \in E_{||f|| > ||\hat{f}||}}\Phi_e(f(e))(||\hat{f}(e)||_{1} - ||f(e)||_{1}) + \\
\phantom{a} & & \frac{1}{\gamma(\mathcal{C})}\sum_{e \in E_{||f|| \leq ||\hat{f}||}}\Phi_e(f(e))(||\hat{f}(e)||_{1}-||f(e)||_{1}) \\
\phantom{a} &\geq& \frac{1}{\alpha(\mathcal{C})}\sum_{e \in E_{||f|| > ||\hat{f}||}}\Phi_e(f(e))||f(e)||_{1} + \\
\phantom{a} & & \frac{1}{\gamma(\mathcal{C})}\sum_{e \in E_{||f|| \leq ||\hat{f}||}}\Phi_e(f(e))||f(e)||_{1} \\
\phantom{a} &\geq& \min\{\frac{1}{\alpha(\mathcal{C})},\frac{1}{\gamma(\mathcal{C})}\}\sum_{e \in E}\Phi_e(f(e))||f(e)||_{1} = \frac{C_{NE}(f)}{\max\{\alpha(\mathcal{C}),\gamma(\mathcal{C})\}}.
\end{eqnarray*}
The first inequality follows from  the definition of the
Pigou bound  applied to each edge in $E_{||f|| > ||\hat{f}||}$ and
the {\em Gamma} bound w.r.t. edges in $E_{||f|| \leq ||\hat{f}||}$.
Here we split the edge set into two cases as mentioned before : $E_{||f|| > ||\hat{f}||}$
and $E_{||f|| \leq ||\hat{f}||}$ indicate a set of edges with $||f(e)||_{1} > ||\hat{f}(e)||_{1}$ and $||f(e)||_{1} \leq ||\hat{f}(e)||_{1}$,
respectively.
The second inequality can be obtained by rearranging terms in the first inequality and the third and the fourth terms in the second inequality
can be combined together and ignored due to the variational inequality characterization presented in \Cref{lem:variational}.
\end{proof}

\subsection{Price of Anarchy is Unbounded for Heterogeneous Polynomial Delays}
\label{sec:unbounded-poa}
In this subsection we show lower bounds on the price of anarchy in %$2$-commodity 
non-atomic networks with heterogeneous polynomial delay functions. 
It is surprising that even in $2$-commodity networks with delay functions that are sums of monomials of single variables,
the price of anarchy depends on the demand requirement. 
% The example in the appendix further illustrates this
% unbounded nature.

We first consider an example of a $2$-commodity network with $2$ edges, illustrated in \Cref{fig:2-comma-nonatomic-general}.
In this network, the top edge $e_1$ is associated with the delay function
$\Phi^i_{e_1}(f_1(e_1),f_2(e_1)) = a_if^{\theta}_1 + f_2$; the bottom edge $e_2$ is associated with the delay function
$\Phi^i_{e_2}(f_1(e_2),f_2(e_2)) = f_1 + a_if^{\theta}_2$ where $a_1, a_2 $ and $\theta >1$.
Demand requirements are defined as $r_1 = r_2 = r$ (where $r \geq 2$).
The worst-case Nash equilibrium flow vector $f$ is achieved when $f_1(e_1) = r$ and $f_2(e_2) = r$ and consequently the cost of Nash Equilibrium, $C_{NE}(f) = (a_1+a_2)r^{\theta+1}$. In contrast the social optimum flow $\hat{f}$ can be obtained from $f_2(e_1) = r$ and $f_1(e_2) = r$ and the social optimum cost  is $C_{SO}(\hat{f}) = 2r^2$. 

% When $a_i=a$ \todo{J: $a_i=a$ is not needed at all right? Without this condition, the bound is $r^{\theta-1}$ so that $\theta=3$ is needed to make $2r^2$} this gives us a bound of $2ar^{\theta-1}$. 
Note that this bound holds for heterogeneous functions that are also decomposable (by letting $a_1=a_2=a$), but does not hold for  decomposable $\theta$-complete polynomials, a class for which upper bounds will be illustrated in \Cref{sec-nonatmoic-poa}.
\begin{lemma}\label{lem:unboundedPoA}
Let $\mathcal{C}$ be a set of polynomial heterogeneous delay functions comprising monomials of a  single variable.
There exists a $(G,R=(r,r),\Phi)$, an instance of non-atomic $2$-commodity network flow routing game where  $\Phi \in \mathcal{C}$, such that
the price of anarchy of 2-commodity flow routing in $(G,R,\Phi)$ is $\Omega(r^{\theta-1})$.
\end{lemma}
\begin{figure}[ht!]
  \begin{center}
  \scalebox{0.55}
  {\includegraphics[width=160mm ,height=35mm]{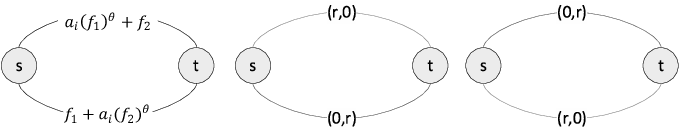}}\\
{\tiny \hspace{0.0in}(a) Network  \hspace{0.45in}(b)Nash Equilibrium \hspace{0.3in}(c) Optimum}
\caption{General Polynomial and Heterogeneous Latency}
  \label{fig:2-comma-nonatomic-general}
  \end{center}
\end{figure}
%\todo{Figure needs some work}
Our next example network in  \Cref{fig:unbounded_price} shows that in fact the inefficiency of equilibrium can be much worse.
%We show an example network in .
\begin{figure}[ht!]
  \begin{center}
  \scalebox{0.8}
  {\includegraphics[width=120mm ,height=20mm]{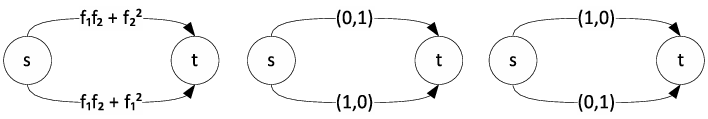}}\\
  {\tiny \hspace*{0.0in}(a) Network  \hspace*{0.75in}(b)Nash Equilibrium \hspace*{0.6in}(c) Optimum}
  \caption{An example of an unbounded price of anarchy}
  \label{fig:unbounded_price}
  \end{center}
\end{figure}
In this network, the top edge $e$ is associated with the delay function
$\Phi_{e}(f_1(e),f_2(e)) = f_1f_2 + f^2_2$ and the bottom edge $h$ is associated with the delay function
$\Phi_{h}(f_1(h),f_2(h)) = f_1f_2 + f^2_1$.
Demand requirements are defined as $r_1 = 1$ and $r_2 = 1$.
The worst-case Nash equilibrium flow vector $f$ is achieved when $f_1(h) = 1$ and $f_2(e) = 1$ and consequently the
Nash equilibrium cost $C_{NE}(f) = 2$. Conversely the social optimum flow $\hat{f}$ can be obtained from $\hat{f}_1(e) = 1$ and $\hat{f}_2(h) = 1$
and the social optimum flow cost $C_{SO}(\hat{f}) = 0$.
Therefore, the price of anarchy is unbounded.

\subsection{PoA for $k$-Commodity Non-atomic Networks: the decomposable $\theta$-complete polynomial case}
\label{sec-nonatmoic-poa}
We next consider the price of anarchy for $k$-commodity non-atomic  network flow routing games
where each  edge is  associated with a decomposable 
$\theta$-complete polynomial
delay function. The lower bounds obtained in the previous section used network instances where the polynomial delay function is not  $\theta$-complete. The first example does not satisfy the requirement that all terms have the same  sum of exponents and the second example  does not have all the non-cross terms corresponding to the two flow variables.

As discussed in the introduction,
it has been shown that Nash equilibrium flow is equivalent to a social optimum flow under homogeneous aggregate delay functions.
However, generally this does not hold, and even for affine delay functions, PoA is bounded by $4/3$.
In this section we establish an upper bound for the PoA  for heterogeneous decomposable $\theta$-complete polynomial delay functions which, unfortunately, is worse than that for the classical aggregate functions used in \cite{ROUGHGARDEN2003341,RT2002}.

\subsubsection{A Lower Bound on PoA}
\label{sec:lower-bound-het}
We consider an example of a $k$-commodity network, with $k$ edges,
illustrated in \Cref{fig:k-comma-poa-non-atomic-poly-hetero}.
In this network, the $i$th edge $e_i$ is associated with the delay function
$\Phi_{e_i}(f_1(e_i),\ldots,f_k(e_i)) = f^{\theta}_1(e_i) + f^{\theta}_2 (e_i)+ \ldots+ a_if^{\theta}_i(e_i)+\ldots + f^{\theta}_k(e_i)$.
Demand requirements are defined as $r_1 = \ldots = r_k = 1$. We let $a_i=a >1 , \forall i$.
The worst-case Nash equilibrium (NE) flow vector $f$ is achieved when $f_j(e_j) = 1,$ when $ 1 \leq j \leq k$ and consequently the
NE cost $C_{NE}(f) = ak$. Conversely the social optimum flow $\hat{f}$ can be obtained by letting $\hat{f}_{\ell}(e_j) = 1/(k-1)$ when $\ell \neq j$, $1 \leq \ell,j \leq k$ for a large enough value of $a$ and the social optimum has cost $C_{SO}(\hat{f}) = k/(k-1)^{\theta-1}$.
Let $a_{max}$ be as defined in \Cref{ModelandDef}. Then we have:
\begin{lemma}
\label{lem:LowerBoundPOA}
Let $\mathcal{C}$ be a set of decomposable $\theta$-complete polynomial delay functions.
There exists a $(G,R,\Phi)$, an instance of a non-atomic $k$-commodity network flow routing game where $\Phi \in \mathcal{C}$, such that
the price of anarchy of 
%the non-atomic $k$-commodity 
flow routing in $(G,R,\Phi)$ is $\Omega(a_{max}(k-1)^{\theta-1})$.
%where $a_{max}=\max_i \{ a_i \} = a$.
\end{lemma}
\begin{figure}[ht!]
  \begin{center}
  \scalebox{0.55}
  {\includegraphics[width=255mm ,height=55mm]{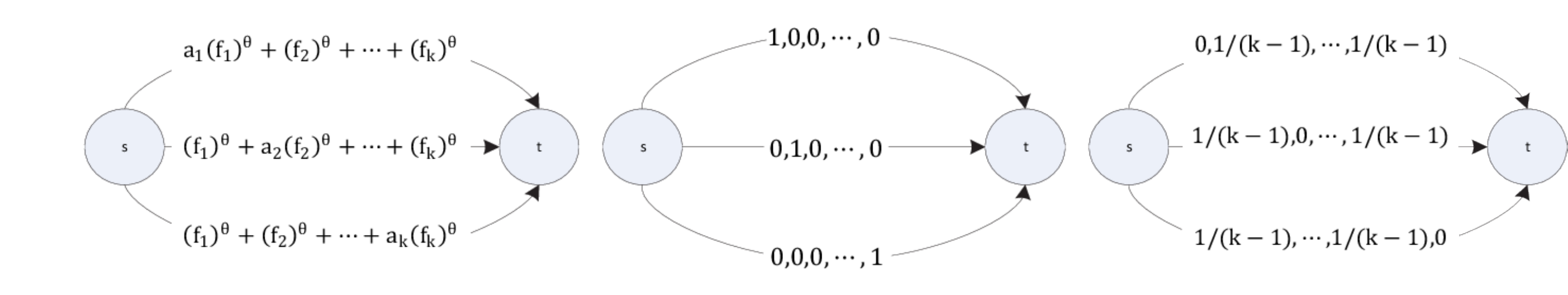}}\\
{\tiny \hspace*{-0.3in}(a) Network  \hspace*{1.35in}(b)Nash Equilibrium \hspace*{0.85in}(c) Optimum}
  \caption{Decomposable $\theta$-complete Polynomial Latency}
  \label{fig:k-comma-poa-non-atomic-poly-hetero}
  \end{center}
\end{figure}

\subsubsection{POA in $k$-commodity non-atomic flow networks with a polynomial decomposable delay function}
In this subsection we consider  network instances where the delay function of each path is in the class of  decomposable
$\theta$-complete polynomial functions. 
% (in short decomposable $\theta$-polynomials).
We let
$\Phi_e(f) = \sum_{\ell}g_{\ell}f^{\theta_1^\ell}_{1}f^{\theta_2^\ell}_{2}\cdots f^{\theta_k^\ell}_{k} + c(e)$ where
$\sum_{j=1}^{j=k} \theta_j^\ell = \theta$ 
%\textcolor{red}{$\leq \theta$} 
and which, by the assumptions of $\theta$-complete polynomials,
includes the terms $\sum_i a_i f_i^\theta, a_i \in \mathbb{R}^{++}$.

In order to prove the price of anarchy we consider two cases depending on the relationship between $||f||_1 $ and $||\hat{f}||_1$, where $f$ and $\hat{f}$ are the Nash and optimum delay flows.  First we consider the case when $||f||_{1} \leq ||\hat{f}||_{1}$ on edge $e$, for which we determine the  {\em Gamma bound}. We will then find the {\em Pigou bound} for the converse inequality.
\begin{lemma}\label{lemma-gamma-bound-nonatomic-poly}
Let $\mathcal{C}$ be a set of decomposable $\theta$-complete polynomial delay functions for non-atomic $k$-commodity network flow routing games.
Then, $\gamma(\mathcal{C}) \leq a_{max}k^{\theta-1}$.
\end{lemma}
\begin{proof}
For the case when $f_1+\ldots+f_{k} \leq \hat{f}_1 + \ldots + \hat{f}_{k}$ on edge $e$, %\todo{Remove  $k'$,Check proof}
w.l.o.g we minimize the social optimum delay by considering $||\hat{f}||_1 = % + \ldots + \hat{f}_{k} = 
||f||_1$. This can be achieved by reducing the flow on each of the dimensions. By monotonicity, $\Phi_e(\hat{f})$ decreases.
For the remainder, we omit the dependency on $e$ for simplicity. Let $r=f$ and $x=\hat{f}$. then
 \begin{eqnarray}
\phantom{a} \gamma(\mathcal{C}) &=& \frac{\Phi(f_1,\ldots,f_k)}{\Phi(\hat{f}_1,\ldots,\hat{f}_{k})} \label{nonatomic-poly-1}\\
\phantom{a} &\leq& \frac{\sum_{\ell}g_{\ell}f^{\theta_1^\ell}_{1}f^{\theta_2^\ell}_{2}\cdots f^{\theta_k^\ell}_{k} + c}{a_1\hat{f}^\theta_1+\ldots+a_{k}\hat{f}^\theta_{k}+c} \label{nonatomic-poly-2}\\
\phantom{a} &\leq& \frac{\max_{\ell}g_{\ell}}{\min_{i}a_i}\frac{f^\theta_1+\ldots+f^\theta_k+f^{\theta_1}_{1}f^{\theta_2}_{2}\cdots f^{\theta_k}_{k}+\ldots}{\hat{f}^\theta_1+\ldots+\hat{f}^\theta_{k}} \label{nonatomic-poly-3}\\
\phantom{a} &<& \frac{\max_{\ell}g_{\ell}}{\min_{i}a_i}\frac{f^\theta_1+\ldots+f^\theta_k+f^{\theta_1}_{1}f^{\theta_2}_{2}\cdots f^{\theta_k}_{k}+\ldots}{k((f_1+\ldots+f_k)/k)^\theta} \label{nonatomic-poly-4}\\
\phantom{a} &\leq& \frac{\max_{\ell}g_{\ell}}{\min_{i}a_i}k^{\theta-1}\frac{f^\theta_1+\ldots+f^\theta_k+f^{\theta_1}_{1}f^{\theta_2}_{2}\cdots f^{\theta_k}_{k}+\ldots}{(f_1+\ldots+f_k)^\theta} \label{nonatomic-poly-5}\\
\phantom{a} &\leq& %\frac{\max_{l}g_{l}}{\min_{l}g_l}
a_{max}k^{\theta-1}\frac{(f_1+\ldots+f_k)^\theta}{(f_1+\ldots+f_k)^\theta} \label{nonatomic-poly-6}\\
\phantom{a} &\leq& a_{max}k^{\theta-1}\label{nonatomic-poly-7}
\end{eqnarray}
In inequality (\ref{nonatomic-poly-2}), $\ell$ represents $\ell$-th term in the delay function $\Phi()$ and $\theta_1 + \ldots + \theta_k = \theta$ since the delay function is a decomposable $\theta$-complete polynomial. To  minimize the divisor we use only the sum of monomials, $\sum_i a_i f_j^{\theta}$, contained in $\Phi(f_1 , \ldots , f_k)$.
This term is lower bounded by $\min_i a_i (
\hat{f}^\theta_1+\ldots+\hat{f}^\theta_{k})$. Since $\hat{f}_1 + \ldots + \hat{f}_{k} = ||f||_1$ the minimum value of the divisor is achieved when $\hat{f}_1 = \ldots = \hat{f}_{k} > ||f||_1/k$.
Finally, (\ref{nonatomic-poly-7}) can be obtained.
\end{proof}

Now let us evaluate the Pigou bound for the case when $||f||_{1} > ||\hat{f}||_{1}$ on an edge $e$.
\begin{lemma}\label{lemma-poa-tc-p}
Let $\mathcal{C}$ be a set of decomposable $\theta$-complete polynomial delay functions for non-atomic $k$-commodity network flow routing games.
% Let $\mathcal{C}$ be a set of %\todo{Change to any theta -not}
% decomposable $\theta$-complete polynomial delay functions and $\alpha(\mathcal{C})$ the Pigou bound for $\mathcal{C}$.
% Let $(G,R,\Phi)$ be an instance of a non-atomic $k$-commodity flow routing game such
% that $\Phi \in \mathcal{C}$.
Then, 
%the  Pigou bound is bounded as: 
$\alpha(\mathcal{C})\leq (\theta+1)\gamma(\mathcal{C})$.
\end{lemma}
\begin{proof}
We consider the value of the Pigou bound and fixing flow $f$, find a bound on the infimum of the denominator,  $\beta(\mathcal{C})$, that occurs in the expression for $\alpha(\mathcal{C})$. We let $r=f$ and $x= \hat{f}$ and determine $ \inf_{\hat{f}} \beta(\mathcal{C})$ in terms of the vector $f$.  Again we omit the dependency on $e$.

Given a requirement vector $R$ we find a lower bound on $\beta(\mathcal{C})$ as  described above.
The function $\beta(\mathcal{C}) $ is convex and thus by KKT conditions w.l.o.g assume that
$ \hat{f}_i >0$ and $\frac{\partial\beta(\mathcal{C})}{\partial \hat{f}_{i}} = 0$, $\forall i \in \{1,2,\ldots,k'\}$,
where $|\{i \in K | \hat{f}_i >0 \}|=k'$. 
Note thus that by the equations, $\frac{\partial \beta(\mathcal{C})}{\partial \hat{f}_i} = 0$ for $i, 1 \leq i \leq k'$, we obtain
\[  \gamma(\mathcal{C})||\hat{f}||_1 \frac{\partial\Phi(\hat{f})}{\partial \hat{f}_i} - \Phi(f) + \gamma(\mathcal{C})\Phi(\hat{f}) = 0\]
which can be written as
\[ \gamma(\mathcal{C})||\hat{f}||_1 \frac{\partial\Phi(\hat{f})}{\partial \hat{f}_i} = \Phi(f) - \gamma(\mathcal{C})\Phi(\hat{f}).\]
Multiplying the above with $\hat{f}_i$ gives that
\[ \gamma(\mathcal{C})||\hat{f}||_1\hat{f}_i \frac{\partial\Phi(\hat{f})}{\partial \hat{f}_i} = \hat{f}_i(\Phi(f) - \gamma(\mathcal{C})\Phi(\hat{f})).\]
and summing over all $1 \leq i \leq k'$,
\[ \gamma(\mathcal{C})\sum^{k'}_{i=1}||\hat{f}||_1\hat{f}_i \frac{\partial\Phi(\hat{f})}{\partial \hat{f}_i} = \sum^{k'}_{i=1}\hat{f}_i(\Phi(f) - \gamma(\mathcal{C})\Phi(\hat{f}))\]
which is equivalent to
\[ \gamma(\mathcal{C})||\hat{f}||_1\sum^{k'}_{i=1}\hat{f}_i \frac{\partial\Phi(\hat{f})}{\partial \hat{f}_i} = ||\hat{f}||_1(\Phi(f) - \gamma(\mathcal{C})\Phi(\hat{f})). \]
Since $\theta$ is the maximum degree and positive and all terms in $\frac{\partial\Phi(\hat{f})}{\partial \hat{f}_i}$ %obtained in the differential
are nonnegative as all co-efficients are nonnegative (by definition) and $\hat{f}$ has  nonnegative components,%\todo{Done-why??},
%adding all the required terms involving the differential in the LHS of the above equation gives: 
\[ \theta \Phi(\hat{f}) \geq \sum^{k'}_{i=1}\hat{f}_i\partial\Phi(\hat{f})/\partial \hat{f}_i \]
and thus
\[ \gamma(\mathcal{C})||\hat{f}||_1\theta\Phi(\hat{f}) \geq ||\hat{f}||_1(\Phi(f) - \gamma(\mathcal{C})\Phi(\hat{f})). \]
By eliminating $||\hat{f}||_1$ on both sides,
\begin{eqnarray}
\phantom{a} \gamma(\mathcal{C})(\theta+1)\Phi(\hat{f}) \geq \Phi(f). \label{phir-phix}
\end{eqnarray}

\noindent
Further,
\begin{eqnarray}
\phantom{a} \beta(\mathcal{C}) &=& \gamma(\mathcal{C})||\hat{f}||_1\Phi(\hat{f}) + (||f||_1 - ||\hat{f}||_1)\Phi(f) \\
\phantom{a} &\geq& ||f||_1\Phi(f) + ||\hat{f}||_1\Phi(f)\frac{\gamma(\mathcal{C})}{(\theta+1)\gamma(\mathcal{C})} -  ||\hat{f}||_1\Phi(f) \\
\phantom{a} &=& ||f||_1\Phi(f) + ||\hat{f}||_1\frac{- \theta\Phi(f)}{\theta+1} \\
\phantom{a} &\geq& ||f||_1\Phi(f) + ||f||_1\frac{- \theta\Phi(f)}{\theta+1} \\
\phantom{a} &=& ||f||_1(\Phi(f) - \frac{\theta}{\theta+1}\Phi(f)) \\
\phantom{a} \implies \beta(\mathcal{C}) &\geq& \frac{1}{\theta+1}||f||_1\Phi(f)
\end{eqnarray}

\noindent
Thus,
\begin{eqnarray}
\phantom{a} \alpha(\mathcal{C}) &\leq& \frac{\gamma(\mathcal{C})||f||_1\Phi(f)}{\gamma(\mathcal{C})||\hat{f}||_1\Phi(\hat{f}) + (||f||_1 - ||\hat{f}||_1)\Phi(f)} \\
\phantom{a} &=& \frac{\gamma(\mathcal{C})||f||_1\Phi(f)}{\beta(\mathcal{C})} \\
\phantom{a} &\leq& \frac{\gamma(\mathcal{C})||f||_1\Phi(f)}{\frac{||f||_1\Phi(f)}{\theta+1}} \\
\phantom{a} &=& (\theta+1)\gamma(\mathcal{C})
\end{eqnarray}
\end{proof}

Combining the above two bounds and using \Cref{thm:poa-commodity-affine} we obtain the following result:
\begin{theorem}\label{thm-2-nonatomic-decomposable-poly}
Let $\mathcal{C}$ be a set of decomposable $\theta$-complete polynomial delay functions.
Let $(G,R,\Phi)$ be a non-atomic $k$-commodity network flow routing game instance with delay function $\Phi \in \mathcal{C}$.
The price of anarchy of flow routing in $(G,R,\Phi)$ is
at most $(\theta+1)a_{max}k^{\theta-1}$
\end{theorem}

\subsection{POA in $k$-commodities non-atomic heterogeneous $\theta$-complete polynomial delay functions}\label{subsec-poa-hetero-poly}
In this section we show bounds on the price of anarchy in $k$-commodity non-atomic networks for heterogeneous
affine and polynomial delay functions.
The bound will be a  function of the number of commodities $k$ and $N_{\Phi}$, the maximum number of terms in the polynomial representing $\Phi^i_e(f)$ over all edges, $e$, and commodities, $i$. Interestingly, there is no dependence on the size of the network.

Before proving the result we will need the following lemma:
\begin{lemma}\label{heterogeneous-poa-poly}
Let $f$ be a Nash equilibrium flow and $\hat{f}$ be a social optimum
%any feasible 
flow vector in a non-atomic $k$-commodity flow routing game instance with heterogeneous $\theta$-complete polynomial delay functions. Then
\[ C_{NE}(f) \leq kN_{\Phi}a_{max}(C_{NE}(f))^{\frac{\theta}{\theta+1}}(C_{SO}(\hat{f}))^{\frac{1}{\theta+1}} + C_{SO}(\hat{f}) \]%\todo{Why CSO at the end}
\end{lemma}
\begin{proof}
We use the variational inequality defined in \Cref{lem:variational}(also see \cite{D1980}).
The delay function may have two categories of terms : i) summation of all terms that consists of one variable  and ii) summation of all terms that consists of at least two variables (referred to as cross terms).
Let us first consider the part of the delay function that comprises terms of category (i).
To simplify the presentation we first  consider one such term in $\Phi^i_e(f)$, i.e., the generic term 
$ a_{i\ell}(e)f^{\theta}_\ell(e)$. A bound on the contribution of the term to the delay is provided below: 
\begin{eqnarray*}
\phantom{a} \sum_e\hat{f}_i(e)(a_{i\ell}(e)f^{\theta}_\ell(e)) &\leq&
a_{max}\sum_e\hat{f}_i(e)(f^{\theta}_\ell(e)) \\
&\leq&
a_{max}(\sum_e\hat{f}^{\theta+1}_i(e))^{\frac{1}{\theta+1}}(\sum_e(f^{\theta}_\ell(e))^{\frac{\theta+1}{\theta}})^{\frac{\theta}{\theta+1}} \\
\phantom{a} &\leq& a_{max}C_{SO}(\hat{f})^{\frac{1}{\theta+1}} \times (\sum_e(f^{\theta}_\ell(e))^{\frac{1}{\theta}}(\Phi^{\ell}_e(f)))^{\frac{\theta}{\theta+1}} \\
\phantom{a} &\leq& a_{max}C_{SO}(\hat{f})^{\frac{1}{\theta+1}} \times (\sum_e(f_\ell(e))(\Phi^{\ell}_e(f)))^{\frac{\theta}{\theta+1}} \\
\phantom{a} &\leq& a_{max}C_{SO}(\hat{f})^{\frac{1}{\theta+1}}C_{NE}(f)^{\frac{\theta}{\theta+1}}
\end{eqnarray*}
The term on the LHS in the first inequality arises on the RHS of  the variational inequality.
The first inequality holds since $a_{max}$ is the largest coefficient over all $a_{ij}$.
The second inequality can be obtained by applying H$\ddot{o}$lder's inequality.
 The third inequality is true since $\forall l$, %\todo{\textcolor{red}{Check!!}}
 $f^{\theta}_\ell(e) \leq \Phi^{\ell}_e(f)$, 
% (recall that all co-efficients are assumed to be $\geq 1)$, 
given the assumption that all co-efficients of the polynomial $\Phi^{\ell}_e(f)$ have value  greater than or equal to one. %\todo{Justify --need stronger--all co-efficients >=1-justify that in the introduction--multiple all co-efficients}
Finally, the last inequality follows since  $\sum_e(f_{\ell}(e))(\Phi^{\ell}_e(f)) \leq C_{NE}(f)$.

For the second case, note that for any cross term the following is true: $f^{\theta_1}_{1}(e)f^{\theta_2}_{2}(e) \cdots f^{\theta_k}_{k}(e) \leq \max(f^{\theta}_1(e), \ldots, f^{\theta}_k(e))$ where $\theta_1 + \ldots + \theta_k = \theta$ in the second case. The $\ell$-th term in the delay function for commodity $i$ satisfies $g_{i\ell}(e)f^{\theta_1^\ell}_{1}(e)f^{\theta_2^\ell}_{2}(e)\cdots f^{\theta_k^\ell}_{k}(e)$ $\leq g_{i\ell}(e)\max(f^{\theta}_1(e), \ldots, f^{\theta}_k(e))$. The subsequent bound in this case is obtained similarly to the first case.
\ignore{
According to the first case, we may have that
\begin{eqnarray*}
\phantom{a} \sum_e\hat{f}_i(e)(g_{i\ell}(e)f^{\theta}_{\ell'}(e)) &\leq&
a_{max}\sum_e\hat{f}_i(e)(f^{\theta}_{\ell'}(e)) \\
&\leq&
a_{max}(\sum_e\hat{f}^{\theta+1}_i(e))^{\frac{1}{\theta+1}}(\sum_e(f^{\theta}_{\ell'}(e))^{\frac{\theta+1}{\theta}})^{\frac{\theta}{\theta+1}} \\
\phantom{a} &\leq& a_{max}C_{SO}(\hat{f})^{\frac{1}{\theta+1}} \times (\sum_e(f^{\theta}_{\ell'}(e))^{\frac{1}{\theta}}(\Phi^{\ell'}_e(f)))^{\frac{\theta}{\theta+1}} \\
\phantom{a} &\leq& a_{max}C_{SO}(\hat{f})^{\frac{1}{\theta+1}} \times (\sum_e(f_{\ell'}(e))(\Phi^{\ell'}_e(f)))^{\frac{\theta}{\theta+1}} \\
\phantom{a} &\leq& a_{max}C_{SO}(\hat{f})^{\frac{1}{\theta+1}}C_{NE}(f)^{\frac{\theta}{\theta+1}}
\end{eqnarray*}
where $\ell' = \arg_{k' \in K}\max(f^{\theta}_1(e), \ldots, f^{\theta}_k(e))$.
Note that the second case proof is almost same as the first case except that we replace $\ell$ with $\ell'$.
}
Based on these two cases, we know that any term in $\Phi^i_e(f)$ has the same upper bound, i.e., $a_{max}C_{SO}(\hat{f})^{\frac{1}{\theta+1}}C_{NE}(f)^{\frac{\theta}{\theta+1}}$.
Note that we have at most $N_{\Phi}$ terms in $\Phi^i_e(f)$ for each of the $k$ commodities. The bound in the lemma follows, since
\begin{eqnarray*}
\phantom{a} C_{NE}(f) &=&   
\sum_i \sum_e f_i(e)\Phi^i_e(f)\\
&=& \sum_{i \in K} \sum_{e \in E} f_i(e)\left( \sum_j a_{ij}(e) f_j^\theta+ 
\sum_{ 1 \leq \ell \leq L^i_e}g_{i\ell}(e)f^{\theta_1^\ell}_{1}f^{\theta_2^\ell}_{2}\cdots f^{\theta_k^\ell}_{k} + c_i(e) \right)\\
&\leq& \sum_i \sum_e \hat{f}_i(e) \left( \sum_j a_{ij}(e) f_j^\theta+ 
\sum_{ 1 \leq \ell \leq L^i_e}g_{i\ell}(e)f^{\theta_1^\ell}_{1}f^{\theta_2^\ell}_{2}\cdots f^{\theta_k^\ell}_{k} \right) + \sum_i\sum_e \hat{f}_i(e)c_i(e)\\
%g_{i\ell}(e)f^{\theta}_\ell(e)) \\
&\leq&
kN_{\Phi}a_{max}(C_{NE}(f))^{\frac{\theta}{\theta+1}}(C_{SO}(\hat{f}))^{\frac{1}{\theta+1}} + C_{SO}(\hat{f}) \\
%a_{max}\sum_e\hat{f}_i(e)(f^{\theta}_\ell(e))\\
\end{eqnarray*}
where the last inequality follows because 
$\sum_i\sum_e \hat{f}_i(e)c_i(e) \leq \sum_i\sum_e \hat{f}_i(e) \Phi^i_e(\hat{f}) = C_{SO}(\hat{f})$.
\end{proof}
We use the above lemma to prove the following result:
\begin{theorem}\label{thm-k-nonatomic-heterogeneous-poly}
Let $\mathcal{C}$ be a set of heterogeneous $\theta$-complete polynomial delay functions.
If $(G,R,\Phi)$ is a  non-atomic $k$-commodity network flow routing game instance with delay functions in $\mathcal{C}$, then the price of anarchy of routing flow in $(G,R,\Phi)$ is at most $(a_{max}kN_{\Phi}+1)^{\theta+1}$.
\end{theorem}
\begin{proof}
% We start with the variational inequality. \todo{J: we do not need to start with the variational inequality. actually we do not need this statement and the following ineqaulity. right? We need to start with Due to Lema...}
% \[C_{NE}(f) \leq \sum_{e \in E}(\sum_{\ell}g_{i\ell}(e)f^{\theta_1\ell}_{1}(e)f^{\theta_2\ell}_{2}(e)\cdots f^{\theta_k\ell}_{k}(e) + c(e))(\sum_i\hat{f}_i(e)) \]
% where $\theta_{1,\ell(e)} + \ldots + \theta_{k,\ell(e)} = \theta$.
Due to \Cref{heterogeneous-poa-poly},
\[C_{NE}(f) \leq kN_{\Phi}a_{max} \left(C_{NE}(f)\right)^{\frac{\theta}{\theta+1}} \left(C_{SO}(\hat{f})\right)^{\frac{1}{\theta+1}} + C_{SO}(\hat{f}) \]
where $f$ is the Nash equilibrium flow and $\hat{f}$ the social optimum flow. Thus,
\[\left(\frac{C_{NE}(f)}{C_{SO}(\hat{f})}\right) \leq kN_{\Phi}a_{max} \left(\frac{C_{NE}(f)}{C_{SO}(\hat{f})}\right)^{\frac{\theta}{\theta+1}} + 1. \]
Since $(\frac{C_{NE}(f)}{C_{SO}(\hat{f})})^{\frac{\theta}{\theta+1}} \geq 1 $,
\[\left(\frac{C_{NE}(f)}{C_{SO}(\hat{f})}\right) \leq (kN_{\Phi}a_{max}+1) \left(\frac{C_{NE}(f)}{C_{SO}(\hat{f})}\right)^{\frac{\theta}{\theta+1}} \]
% Let $x = \left(\frac{C_{NE}(f)}{C_{SO}(\hat{f})}\right)^{\frac{1}{\theta+1}}$ and $t = kN_{\Phi}a_{max}$.
% A solution to the above can be obtained from the root of $F(x): x^{\theta+1} - tx^\theta - 1$. A bound  is shown in lemma \ref {xn-converge-to-2} in the appendix. %\ref{appenxix-k-non}.
% Then, $x \leq 2kN_{\Phi}a_{max}$, and 
Thus   $PoA \leq (a_{max}kN_{\Phi}+1)^{\theta+1}$. 
\end{proof}

\subsection{$k$-commodity non-atomic - heterogeneous affine delay functions}\label{subsec-poa-hetero-affine}
We determine a bound on the price of anarchy in $k$-commodity non-atomic networks with heterogeneous affine delay functions. %\todo{Just require a corollary to the Theorem 2.10}
From \Cref{heterogeneous-poa-poly}, we have the following corollary:
\begin{corollary}\label{general-poa-affine}
Let $f$ and $\hat{f}$ be a Nash equilibrium flow and a social optimum %any feasible 
flow vector, respectively,
in a non-atomic $k$-commodity network flow routing game instance with heterogeneous affine delay functions. Then
\[ C_{NE}(f) \leq kN_{\Phi}a_{max}\left( C_{NE}(f)C_{SO}(\hat{f}) \right)^{1/2} + C_{SO}(\hat{f}). \]
\end{corollary}
Also, in a fashion similar to \Cref{thm-k-nonatomic-heterogeneous-poly}, we have the following theorem.
%\todo{Where are the details}
\begin{theorem}\label{thm-k-nonatomic-heterogeneous-affine}
Let $\mathcal{C}$ be a set of heterogeneous affine functions.
If $(G,R,\Phi)$ is a $k$-commodity non-atomic network flow routing game instance with delay function 
in $\mathcal{C}$, then the price of anarchy of routing flow in $(G,R,\Phi)$ is at most 
$k^2N^2_{\Phi}a^2_{max}+2$.
%\todo{Changed-is proof there?}
% $(k^2N^2_{\Phi}a^2_{max}+\sqrt{(k^2N^2_{\Phi}a^2_{max})(k^2N^2_{\Phi}a^2_{max}+4)})/2 + 1$.
\end{theorem}
\begin{proof}
Using \Cref{general-poa-affine}
\[\left(\frac{C_{NE}(f)}{C_{SO}(\hat{f})}\right) \leq (kN_{\Phi}a_{max}) \left(\frac{C_{NE}(f)}{C_{SO}(\hat{f})}\right)^{\frac{1}{2}} +1\]
The result follows from determining an upper bound on $x$ such that $x^2 - kN_{\Phi}a_{max}x -1 \leq 0$ where $x^2=\frac{C_{NE}(f)}{C_{SO}(\hat{f})}$.

\end{proof}

%%%%%%%%%%%%%%%%%%%%%%%%%%%%%%%%%%%%%%%%%%%%%%% poa - atomic
\section{Part II - PoA for Atomic Flows}
\label{secAtomic}
In this section we consider atomic flows. We start with some negative  results on the existence of Nash equilibrium.
In general, Nash equilibrium does not exist, even for affine functions, as shown in
\Cref{subsec:ExistenceAtomic}. 
Note that this is somewhat surprising since it is well known that a congestion game with
delay that is a function of aggregate load, or flow, has at least one pure Nash equilibrium \cite{R1973, MS1996}. A positive outlook does exist:
we show that Nash equilibrium exists when the heterogeneous delay functions are affine and uniform over the network.
%\todo{change reference 27 to 3}
Our contributions for the price of anarchy in atomic selfish routing problems have been illustrated in \Cref{table:k}.
 For $2$-commodity networks with uniform affine delay functions, we show (also see \Cref{table:2}) 
 improved bounds which are different from the bounds on general decomposable affine delay functions.
The bound provided (\Cref{thm-2-atomic-uniform-affine}) is an asymptotically tight bound as illustrated by the lower bound of  $\Omega(\sqrt{a_{max}})$.

\subsection{Existence of Atomic NE flow?}
\label{subsec:ExistenceAtomic}
We first show that pure Nash equilibrium does not exist in the case of atomic, decomposable and affine delay functions with unweighted demand requirement.  \begin{figure}[ht!]
  \begin{center}
  \scalebox{0.75}
  {\includegraphics[width=98mm ,height=30mm]{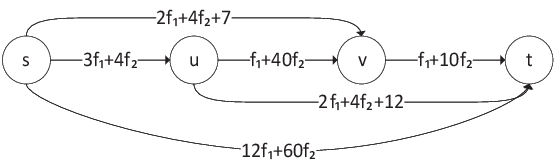}}
  \caption{Affine and Decomposable Latency}
  \label{fig:non-existence-atomic-decomposable}
  \end{center}
\end{figure}
In \Cref{fig:non-existence-atomic-decomposable}, there are two users with $r_1 = r_2 = 1$(also termed unweighted) and
the users have  the same set of strategies (paths) from source $s$ to destination $t$.
The delay function associated with an edge is shown in \Cref{fig:non-existence-atomic-decomposable}.
Let us define four paths $P_1 = \{(s,t)\}, P_2 = \{(s,u), (u,t)\}, P_3 = \{(s,u), (u,v), (v,t)\}$ and $P_4 = \{(s,v), (v,t)\}$.
Since the coefficient  of the delay for user 2  is very large for the edges in $P_1$ and $P_3$, user 2 will not utilize these paths.
Thus, there are a total of  eight pairs of paths each of which are chosen by user 1 and user 2, respectively as shown in \Cref{table:k-atomic-no-pne}.
\begin{table}[ht]
\begin{center}
\scalebox{0.7}{
\begin{tabular}[width=\textwidth]{|c|cccc|}
\hline
 & \multicolumn{2}{|c|}{path} & \multicolumn{2}{|c|}{delay} \\ [-1ex]
\raisebox{2.0ex}{Strategy Pair} & \multicolumn{1}{|c}{user 1} & \multicolumn{1}{c|}{user 2} & \multicolumn{1}{|c}{user 1} & \multicolumn{1}{c|}{user 2} \\ [1ex]
\hline
\hline
\multicolumn{1}{|c|}{$(P_3,P_2)$} & \multicolumn{1}{c|}{$P_3$} & \multicolumn{1}{c|}{$P_2$} & \multicolumn{1}{c|}{$9$} & \multicolumn{1}{c|}{$23$} \\
\hline
\multicolumn{1}{|c|}{$(P_3,P_4)$} & \multicolumn{1}{c|}{$P_3$} & \multicolumn{1}{c|}{$P_4$} & \multicolumn{1}{c|}{$15$} & \multicolumn{1}{c|}{$22$} \\
\hline
\multicolumn{1}{|c|}{$(P_1,P_4)$} & \multicolumn{1}{c|}{$P_1$} & \multicolumn{1}{c|}{$P_4$} & \multicolumn{1}{c|}{$12$} & \multicolumn{1}{c|}{$21$} \\
\hline
\multicolumn{1}{|c|}{$(P_1,P_2)$} & \multicolumn{1}{c|}{$P_1$} & \multicolumn{1}{c|}{$P_2$} & \multicolumn{1}{c|}{$12$} & \multicolumn{1}{c|}{$20$} \\
\hline
\multicolumn{1}{|c|}{$(P_4,P_2)$} & \multicolumn{1}{c|}{$P_4$} & \multicolumn{1}{c|}{$P_2$} & \multicolumn{1}{c|}{$10$} & \multicolumn{1}{c|}{$20$} \\
\hline
\multicolumn{1}{|c|}{$(P_2,P_4)$} & \multicolumn{1}{c|}{$P_2$} & \multicolumn{1}{c|}{$P_4$} & \multicolumn{1}{c|}{$17$} & \multicolumn{1}{c|}{$21$} \\
\hline
\multicolumn{1}{|c|}{$(P_2,P_2)$} & \multicolumn{1}{c|}{$P_2$} & \multicolumn{1}{c|}{$P_2$} & \multicolumn{1}{c|}{$25$} & \multicolumn{1}{c|}{$25$} \\
\hline
\multicolumn{1}{|c|}{$(P_4,P_4)$} & \multicolumn{1}{c|}{$P_4$} & \multicolumn{1}{c|}{$P_4$} & \multicolumn{1}{c|}{$24$} & \multicolumn{1}{c|}{$24$} \\
\hline
\end{tabular}}
\caption{Non-existence of Pure Nash Equilibrium}
\label{table:k-atomic-no-pne}
\end{center}
\end{table}
%\todo{what is the point of the first column}
The strategies of the two users will be represented by the pair of strategies $(P,P')$ where $P$ and $P'$
belong to the set of 4 paths $\{ P_1, P_2, P_3, P_4)$.
There exists a sequence of strategy pairs that cycle \{$(P_3,P_2), (P_3,P_4), (P_1,P_4), (P_1,P_2)$\}.
Further, other strategies are not stable.
Strategy pairs $(P_4,P_2)$ and $(P_2,P_2)$ will shift to the strategy pair $(P_3,P_2)$ because $P_3$
is more beneficial to user 1.
And lastly, strategy pairs $(P_2,P_4)$ and $(P_4,P_4)$ are unstable since the strategy pair
$(P_3,P_4)$ is preferred by user 1.
This example shows that there may be no pure Nash equilibrium in $2$-commodity atomic networks where each edge is associated with an affine, decomposable delay function and each user has a requirement of one unit of demand
(a case usually referred to as unweighted).
% Note that this is somewhat surprising since it well known that any unweighted congestion game with homogeneous
% delays has at least one pure Nash equilibrium \cite{R1973, MS1996}.

However, in  atomic networks when each edge is associated with an affine, uniform delay function, we show that there exists at least one pure Nash equilibrium. This proof applies to both unweighted and weighted demands and is based on the existence of a potential function.
\begin{theorem}
Let $\mathcal{C}$ be a set of affine, uniform delay functions.
If $(G,R,\Phi)$ is a $k$-commodity atomic network flow routing game instance with delay functions in $\mathcal{C}$, then $(G,R,\Phi)$ admits at least one Nash equilibrium.
\end{theorem}
\begin{proof}
We start with considering delay functions of the form $ \sum_i a_i(e) f_i(e) + c(e)$. 
%\todo{change $a_i(e)$ to $a_i$ as it is uniform case}
Let us define a potential function $\Psi(f) = \sum_e(d^2(e)+\sum_{j \in S(e)}(a_j(e)r_j(e))^2)$ where $d(e) = \sum^k_{i=1}a_i(e)r_i(e)+c(e)$ for every feasible flow $f$ and $S(e)$ represents the set commodities utilizing edge $e$.
We claim that a global minimum $f$ of the potential function $\Psi$ is also an equilibrium flow for $(G,R,\Phi)$.
Let $P_i$ be the path used by commodity $i$ in the solution that minimizes $\Psi(f)$.
Assume, for a contradiction, that shifting from path $P_i$ to path $\bar{P}_i$  by user $i$, creating the flow $\bar{f}$, strictly decreases its delay.
In other words,
\begin{eqnarray*}
\phantom{a} \Delta(f,\bar{f}) &=& \Phi_{\bar{P}_i}(\bar{f}) - \Phi_{P_i}(f) \\
\phantom{a} &=& \sum_{e \in \bar{P}_i \setminus P_i}\Phi_e(\ldots, \bar{f}_{i}(e)=r_i, \ldots ) - \sum_{e \in P_i \setminus \bar{P}_i}\Phi_e(f_1(e), \ldots,f_i(e)=r_i, \ldots f_k(e)) < 0.
\end{eqnarray*}
%\todo{Check $f_i(e)=0$ above}
On the other hand, let us consider the potential function $\Psi$ as  user $i$ deviates.
For any edge $\bar{e}$ in $\bar{P}_i \setminus P_i$, we gain $(d(\bar{e}) + a_i(\bar{e})r_i(\bar{e}))^2 + a^2_i(\bar{e})r^2_i(\bar{e}) - d^2(\bar{e}) = 2d(\bar{e})a_i(\bar{e})r_i(\bar{e}) + 2a^2_i(\bar{e})r^2_i(\bar{e})$.
For any edge $e$ in $P_i \setminus \bar{P}_i$, we have a change of  $d^2(e) - ((d(e) - a_i(e)r_i(e))^2 - a^2_i(e)r^2_i(e)) =$  $2d(e)a_i(e)r_i(e)$.
Thus,
\begin{eqnarray*}
\phantom{a} \Psi(\bar{f}) - \Psi(f) &=& \sum_{\bar{e}}(2d(\bar{e})a_i(\bar{e})r_i(\bar{e}) + 2a^2_i(\bar{e})r^2_i(\bar{e})) - \sum_{e}2d(e)a_i(e)r_i(e) \\
&=& 2a_ir_i(\sum_{\bar{e}}(d(\bar{e}) + a_ir_i) - \sum_{e}d(e)) \\
&=& 2a_ir_i(\Phi_{\bar{P}_i}(\bar{f}) - \Phi_{P_i}(f)) \\
&=& 2a_ir_i\Delta(f,\bar{f})
\end{eqnarray*}
The first equality holds due to the definition of potential function $\Psi$.
Since we consider affine, uniform delay functions,  $a_i(\bar{e}) = a_i(e) = a_i$ and this results in the second equality.
Note that the third equality is satisfied because $a_ir_i$ is the exact amount of delay increase after the user $i$ deviation.
Since the delay decreases after user $i$ deviates, $\Delta(f,\bar{f})$ is negative.
Further, $a_i,r_i > 0$, and thus the potential function value at $\bar{f}$ is strictly less than the potential value at $f$, which contradicts that $f$ is a global minimum.
\end{proof}

\subsection{$k$-commodity atomic decomposable $\theta$-complete polynomial delay functions}\label{subsec-poa-general-atomic-poly}
In this section we provide bounds on the {\em price of anarchy} in $k$-commodity network flows with decomposable delay functions.

\subsubsection{A Lower Bound on PoA}
We show an example network in \Cref{fig:2-comma-poa-coefficient}.
\begin{figure}[ht!]
  \begin{center}
  \scalebox{0.55}
  {\includegraphics[width=155mm ,height=30mm]{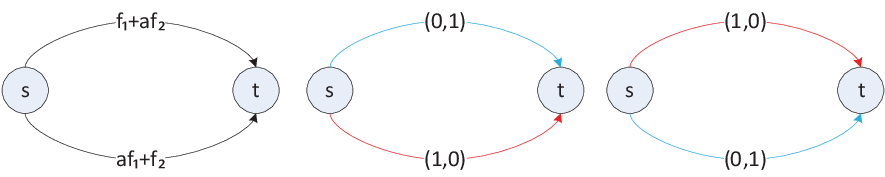}}\\
{\tiny \hspace*{0.0in}(a) Network  \hspace*{0.55in}(b)Nash Equilibrium \hspace*{0.4in}(c) Optimum}

  \caption{Affine and Decomposable Latency}
  \label{fig:2-comma-poa-coefficient}
  \end{center}
\end{figure}
In this network, the top edge $e$ is associated with the delay function
$\Phi_{e}(f_1(e),f_2(e)) = f_1 + af_2$ and the bottom edge $h$ is associated with the delay function
$\Phi_{h}(f_1(h),f_2(h)) = af_1 + f_2$.
Demand requirements are defined as $r_1 = 1$ and $r_2 = 1$.
The worst-case Nash equilibrium flow vector $f$ is achieved when $f_1(h) = 1$ and $f_2(e) = 1$ and consequently the
Nash equilibrium cost $C_{NE}(f) = 2a$. Conversely the social optimum flow $\hat{f}$ can be obtained from $\hat{f}_1(e) = 1$ and $\hat{f}_2(h) = 1$
and the social optimum cost $C_{SO}(\hat{f}) = 2$.
\begin{lemma}
\label{lem:LBAtomic}
Let $\mathcal{C}$ be a set of affine decomposable delay functions.
There exists $(G,R,\Phi)$, an atomic two-commodity network flow routing game instance where $\Phi \in \mathcal{C}$, such that the price of anarchy of atomic flow routing in $(G,R,\Phi)$ is $\Omega(a_{max})$ when $a_{max} \geq 2$.
\end{lemma}
Note that when $a_{max} = 1$, i.e. in the aggregate model where all co-efficients are the same, an example with 4 players has been shown where the price of anarchy is $1+ \phi \approx 2.618$ \cite{AAE2005}.
%it has been shown that the Price of Anarchy is at most $ 4/3$ \cite{RT2002}.

\subsubsection{An Upper Bound on PoA}
In this section we provide an upper bound on the {\em price of anarchy} for $k$-commodity network atomic flows.
We first consider affine delay functions, i.e., delay functions of the form $\Phi_i(e) = \sum_j a_j(e) f_j(e) + c(e)$. We use the Cauchy-Schwartz inequality to prove the following:
\begin{lemma}
\label{atomic-affine-cauchy}
Let $\mathcal{C}$ be a set of affine decomposable delay functions.
Let $f$ and $\hat{f}$ be a Nash equilibrium and a social optimum atomic flow, respectively, in a $k$-commodity network flow routing game  with delays in the class $\mathcal{C}$. Then
\[ C_{NE}(f) \leq (C_{SO}(\hat{f}))^{1/2}\sqrt{a_{max}}(C_{NE}(f))^{1/2} + C_{SO}(\hat{f}). \]
\end{lemma}
\begin{proof}
Let $f$ and $\hat{f}$ be a Nash equilibrium flow and a social optimum flow, respectively for this instance.
Let $P_i$ and $\hat{P}_i$ be paths, utilized by commodity $i$, in the Nash equilibrium (NE) and  the social optimum solution, respectively.
Then by the fact that $P_i$ is used at the NE solution,
\[ \sum_{e \in P_i}(a_1(e)f_1(e) + \ldots + a_k(e)f_k(e) + c(e)) \leq \sum_{e \in \hat{P}_i}(a_1(e)f_1(e) + \ldots + a_k(e)f_k(e) + c(e) + a_i(e)r_i). \]
The inequality above holds due to the variational inequality.
By multiplying both sides of the above inequality by $r_i$ and summing over all commodities, we have the following.%\todo{Referee complains-why??}
\[ \sum_ir_i\sum_{e \in  P_i}(a_1(e)f_1(e) + \ldots + a_k(e)f_k(e) + c(e)) \leq \sum_ir_i\sum_{e \in \hat{P}_i}(a_1(e)f_1(e) + \ldots + a_k(e)f_k(e) + c(e) + a_i(e)r_i) \]
or equivalently, on summing over all edges instead of paths, using $f_i(e)=r_i$ and interchanging the order of the summations, we get
\[ \sum_{e \in \cup_i P_i}||{f}(e)||_1(a_1(e)f_1(e) + \ldots + a_k(e)f_k(e) + c(e)) \leq \sum_{e \in \cup_i \hat{P}_i}||\hat{f}(e)||_1 (a_1(e)f_1(e) + \ldots + a_k(e)f_k(e) + c(e)) \]
\[ + \sum_i\sum_{e \in \hat{P}_i} a_i(e)\hat{f_i}^2 \].
%where $f_i(e)=r_i$.
% by considering cost on all edges instead of cost on paths. 
% Thus we have
% \begin{eqnarray*}
% \phantom{a} & &\sum_e||f(e)||_1(a_1(e)f_1(e) + \ldots + a_k(e)f_k(e) + c(e)) \\
% \phantom{a} &\leq& \sum_e||\hat{f}(e)||_1(a_1(e)f_1(e) + \ldots + a_k(e)f_k(e) + c(e)) \\
% \phantom{a} &\leq& \sum_e||\hat{f}(e)||_1(a_1(e)f_1(e) + \ldots + a_k(e)f_k(e)) + \\
% \phantom{a} & &\sum_e\left(||\hat{f}(e)||_1c(e) + \sum_i\hat{f}^2_i(e)a_i(e)\right).
% \end{eqnarray*}

Then, cost at Nash equilibrium is bounded above as follows:
\begin{eqnarray*}
\phantom{a} C_{NE}(f) &\leq& \sum_e||\hat{f}(e)||_1(a_1(e)f_1(e) + \ldots + a_k(e)f_k(e)) + \sum_e(||\hat{f}(e)||_1c(e) + \sum_i\hat{f}^2_i(e)a_i(e)) \\
\phantom{a} &\leq& \sum_e||\hat{f}(e)||_1(a_1(e)f_1(e) + \ldots + a_k(e)f_k(e)) + \sum_e||\hat{f}(e)||_1\sum_i(a_i(e)\hat{f}_i(e)+c(e)) \\
\phantom{a} &\leq& (\sum_e||\hat{f}(e)||^2_1)^{1/2}(\sum_e(a_1(e)f_1(e) + \ldots + a_k(e)f_k(e))^2)^{1/2} + \\
\phantom{a} & & \sum_e||\hat{f}(e)||_1\sum_i(a_i(e)\hat{f}_i(e)+c(e)) \\
\phantom{a} &\leq& \sqrt{(\sum_e||\hat{f}(e)||^2_1)(\max_{i,e}a_i(e))\sum_e(f_1(e) + \ldots + f_k(e))(a_1(e)f_1(e) + \ldots + a_k(e)f_k(e))} \\
\phantom{a} & & + \sum_e||\hat{f}(e)||_1\sum_i(a_i(e)\hat{f}_i(e)+c(e)) \\
\phantom{a} &\leq& (C_{SO}(\hat{f}))^{1/2}\sqrt{\max_{i,e}a_i(e)C_{NE}(f)} + C_{SO}(\hat{f})
\end{eqnarray*}
%\todo{In last inequality co-efficients are > 1? integers-Make consistent}
The third inequality can be obtained by utilizing the Cauchy-Schwartz inequality. The last inequality depends on the co-efficients of the affine function being greater than or equal to one. 
\end{proof}
Using the above lemma we can determine the price of anarchy for the set of affine decomposable delay functions.
\begin{theorem}\label{thm-k-atomic-decomposable-affine}
Let $\mathcal{C}$ be a set of affine, decomposable delay functions.
If $(G,R,\Phi)$ is a $k$-commodity atomic network flow routing game instance with the delay function $\Phi$ in $\mathcal{C}$, then the price of anarchy of atomic flow routing in $(G,R,\Phi)$ is at most 
$a_{max}+2$
% $\frac{a_{max}+(a_{max}^2+4a_{max})^{1/2}}{2} + 1$, 
where $a_{max}$ represents the maximum coefficient in the affine delay function.
\end{theorem}
\begin{proof}
From \Cref{atomic-affine-cauchy}
\[ C_{NE}(f) - a_{max}^{1/2}\left( C_{NE}(f)C_{SO}(\hat{f}) \right)^{1/2} - C_{SO}(\hat{f}) \leq 0, \]
and dividing by $C_{SO}(\hat{f})$ gives
\[ \left(\frac{C_{NE}(f)}{C_{SO}(\hat{f})}\right) - a_{max}^{1/2}\left(\frac{C_{NE}(f)}{C_{SO}(\hat{f})}\right)^{1/2} - 1 \leq 0 \]
Then, we obtain a quadratic equation $x^2 - a_{max}^{1/2}x - 1 \leq 0$ where $x = \left(\frac{C_{NE}(f)}{C_{SO}(\hat{f})}\right)^{1/2}$ which leads to PoA = $\left(\frac{(a_{max}+4)^{1/2}+(a_{max})^{1/2}}{2}\right)^2 = \frac{a_{max}+(a_{max}^2+4a_{max})^{1/2}}{2} + 1 \leq a_{max} +2$.
\end{proof}

We can extend the above proof to  find an upper bound on the price of anarchy in $k$-commodity atomic networks with decomposable polynomial delay functions. 
%The proof technique is similar to previous proofs (also see \cite{AAE2005}) and is detailed in the appendix.
\begin{theorem}\label{thm-k-atomic-decomposable-poly}
Let $\mathcal{C}$ be a set of decomposable $\theta$-complete polynomial delay functions.
If $(G,R,\Phi)$ is a $k$-commodity atomic network flow routing game  instance with the delay function  $\Phi$  in $\mathcal{C}$, then the price of anarchy of atomic flow routing in $(G,R,\Phi)$ is 
%$O(a_{max}(N_{\Phi}+k\theta)^{\theta+1})$. or
$O(a_{max}^{\theta+2}N_{\Phi}^{\theta+1})$
\end{theorem}
\begin{proof}
Using the variational inequality applicable in the case of atomic flow routing, we get, for the $i$th commodity:
\begin{equation}
\label{eq:varforlemma-k-atomic}
\sum_{e \in P_i} \Phi_e(f_e) \leq \sum_{e \in P'_i} \Phi_e(\hat{f}_e)%f_e+r_i)
\end{equation}
where $f$ is the Nash equilibrium flow and $\hat{f}$ is the flow function with the $i$th commodity using $P'_i$ instead of $P_i$ to route $r_i$ units of flow.
The introduction of the $i$th commodity to edges in the path $P'_i$ adds a number of terms to the delay. 
%For delay corresponding to the $i$th commodity. 
Consider the following sum of delay terms on edge $e$ when $\hat{f}_i(e)$ is nonzero. 
% $$\sum_{ 1 \leq \ell \leq L^i_e}g_{i\ell}(e)f^{\theta_1^\ell}_{1}f^{\theta_2^\ell}_{2}\cdots f^{\theta_k^\ell}_{k} + c_i(e)$$
$$\sum_{ 1 \leq \ell \leq L^i_e}g_{\ell}(e)\hat{f}^{\theta_1^\ell}_{1}\hat{f}^{\theta_2^\ell}_{2}\cdots \hat{f}^{\theta_k^\ell}_{k} $$% + c_i(e)$$
where $\hat{f}_j(e)=f_j(e), j \neq i$ is nonzero and $\hat{f}_i(e)=r_i$.
Consider one such term $T= g_{\ell}(e)f^{\theta_1^\ell}_{1}f^{\theta_2^\ell}_{2}\cdots f^{\theta_k^\ell}_{k}$.
If $i = \arg \max_{j: e \in P_j } r_j$ then
$T \leq g_{\ell}(e) r_i^{\theta}$ (note that $r_i \in \mathbb{N}$). %\todo{$r_i\geq 1$}
Alternatively,
$T \leq g_{\ell}(e) f_m^{\theta}$ where $m = \arg \max_{j: e \in P_j } r_j$. In both cases,  $T \leq a_{max} \Phi_e(f) $. A similar analysis holds for the terms introduced by the $i$th commodity in the delay of other commodities. Since there are $N_{\Phi}$ terms in the delay function, we get:
\begin{equation}
\label{eq:lemma-k-com-PoA}
%\sum_{e \in P'_i}
\Phi_e(\hat{f}) \leq a_{max}N_{\Phi} \Phi_e(f) +a_{max}N_{\Phi} r_i^{\theta}
\end{equation}
%Using equation (\Cref{eq:varforlemma-k-atomic}), 
Multiplying both sides of \Cref{eq:varforlemma-k-atomic} by $r_i$, applying
\Cref{eq:lemma-k-com-PoA} and summing over all commodities and edges we get:
\begin{equation*}
\begin{aligned}
C_{NE}(f) \phantom{a}&\leq  a_{max}N_{\Phi} \left [
\sum_{e} ||\hat{f}(e)||_1 \Phi_e(f) +  \sum_e \sum_i \hat{f}_i(e)^{\theta +1} \right ] \\        
\phantom{a} &\leq a_{max}N_{\Phi}\left[
(\sum_{e} (||\hat{f}(e)||_1)^{\theta+1})^{\frac{1}{\theta+1}} ( \sum_e \Phi_e(f)^{\frac{\theta+1}{\theta}})^{\frac{\theta}{\theta+1}}  +  \sum_e \sum_i \hat{f}_i(e)^{\theta +1} \right]\\
\phantom{a} &\leq a_{max}N_{\Phi}\left[ 
C_{SO}(\hat{f})^{\frac{1}{\theta+1}} \cdot
( a_{max}^{\frac{1}{\theta}} \sum_e ||f(e)||_1 \Phi_e(f)  )^{\frac{\theta}{\theta+1}}
+ C_{SO}(\hat{f}) \right]
\end{aligned}
\end{equation*}
since
$$ \sum_e \Phi_e(f)^{\frac{\theta+1}{\theta}} \leq  \sum_e \left(  a_{max}(|| f(e) ||_1)^{\theta} \Phi_e(f)^{\theta} \right)^{1/\theta} . $$
Thus
\begin{equation*}
C_{NE}(f)\leq a_{max}N_{\Phi}\left[ 
C_{SO}(\hat{f})^{\frac{1}{\theta+1}} \cdot
( a_{max}^{\frac{1}{\theta}}C_{NE})^{\frac{\theta}{\theta+1}} \right] + C_{SO}(\hat{f})
\end{equation*}
or equivalently
\begin{equation*}
\begin{aligned}
 \frac{C_{NE}(f)}{C_{SO}(\hat{f})} \phantom{a} &\leq 
a_{max}^{\frac{\theta+2}{\theta+1}}N_{\Phi}%^{\frac{2\theta+1}{\theta+1}}
\left( \frac{C_{NE}(f)}{C_{SO}(\hat{f})}\right)^{\frac{\theta}{\theta+1}} +1\\
\phantom{a} &\leq (a_{max}^{\frac{\theta+2}{\theta+1}}N_{\Phi}+1)
\left( \frac{C_{NE}(f)}{C_{SO}(\hat{f})}\right)^{\frac{\theta}{\theta+1}} 
\end{aligned}
\end{equation*}
The bound on the price of anarchy is thus 
$O(a_{max}^{\theta+2}N_{\Phi}^{\theta+1})$.
\end{proof}
%\todo{Proof needed}
\subsection{Improved PoA for %Uniform Two-Commodity Networks, 
Affine Delays}\label{sec-poa-2com-atomic}
In this subsection we provide improved bounds for  atomic flows when the delay function
is affine and uniform modulo differences  in the constant term, i.e., $\Phi(e) = \sum_i a_i f_i(e) + c(e)$.
%, i.e. all edges in the network have the same delay function. 
% We first show an example which establishes a lower bound on the behavior of the price of anarchy.
%\subsubsection{Lower Bounds on PoA for Uniform Delay Functions}
% The proof of Lemma \ref{lower-bound-affine-uniform-atomic} is shown below.
\subsubsection{Lower Bounds on PoA for Uniform Delay Functions}
In fact, even for the class of affine functions that does not have a constant term, we provide an example (\Cref{fig:comm-dependent-poa1}) which shows that the PoA is $\Omega(\sqrt{a})$.
In this example the delay function associated with the edges
is $af_1 + f_2 + af_3$. The requirements are $r_1=r_3=1$ and $r_2 = \sqrt{a}$.
  \begin{figure}[ht]
   \vspace*{-2.25in}
  \begin{center}
  \scalebox{0.65}
%   {\includegraphics[width=175mm ,height=70mm]{PoAnewfig.eps}}
    {\includegraphics[width=150mm ,height=230mm]{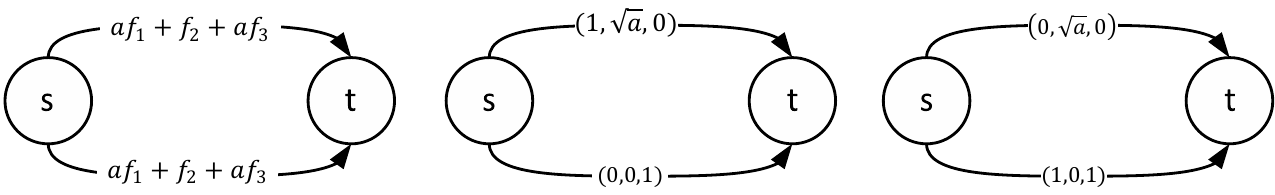}}\\
   \vspace*{-2.25in}
{\tiny \hspace*{0.0in}(a) Network  \hspace*{0.5in}(b)Nash Equilibrium \hspace*{0.35in}(c) Optimum}

  \caption{Uniform and affine latency}
  \label{fig:comm-dependent-poa1}
  \end{center}
  \end{figure}
The Nash equilibrium solution has flow $f_1=1$ and $f_2=\sqrt{a}$ on the top edge. It has flow $f_3=1$ on the bottom edge with cost $C_{NE}(f)= a\sqrt{a} + 3a + \sqrt{a}$. The optimum flow solution is $\hat{f}_1$ and $\hat{f}_3$ on the bottom edge and $\hat{f}_2$ on the top edge with $C_{SO}(\hat{f})=5a$. Thus $POA=\Omega(\sqrt{a})$.

% Another class of functions that has similar lower bounds on the price of anarchy is the class where we allow constants to differ on edges.
For functions of the form  $af_1 + f_2+c(e)$ with $a\geq1$ and $c(e)$ a nonnegative constant, we provide a similar example that shows a PoA that is  $\Omega(\sqrt{a})$.
  \begin{figure}[ht]
  \begin{center}
  \scalebox{0.55}
  {\includegraphics[width=155mm ,height=30mm]{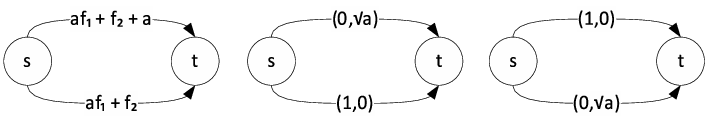}}\\
  {\tiny \hspace*{0.0in}(a) Network  \hspace*{0.5in}(b)Nash Equilibrium \hspace*{0.35in}(c) Optimum}

  \caption{Uniform and affine latency}
  \label{fig:comm-dependent-poa}
  \end{center}
  \end{figure}
In this network, the
top edge $e$ is associated with the delay function
$\Phi_{e}(f_1(e),f_2(e)) = af_1 + f_2 + a$ and the bottom edge $h$ is associated with the delay function
$\Phi_{h}(f_1(h),f_2(h)) = af_1 + f_2$.
Demand requirements are defined as $r_1 = 1$ and $r_2 = (a)^{1/2}$.
The Nash equilibrium(NE) flow vector $f$ has the flow $f_1(h) = 1$ and $f_2(e) = (a)^{1/2}$ and consequently the
cost of NE, $C_{NE}(f) = a((a)^{1/2}+2)$. Conversely the social optimum flow $\hat{f}$
can be obtained from $\hat{f}_1(e) = 1$ and $\hat{f}_2(h) = (a)^{1/2}$
and the social optimum cost $C_{SO}(\hat{f}) = 3a$.
The price of anarchy is thus $\frac{a((a)^{1/2}+2)}{3a} = \frac{(a)^{1/2}+2}{3}$.

\begin{lemma}\label{lower-bound-affine-uniform-atomic}
Let $\mathcal{C}$ be a set of affine delay functions that are uniform (modulo differences in the constant term). %\todo{check definition section}
There exists $(G,R,\Phi)$, an atomic 2-commodity network flow network routing game instance where $\Phi \in \mathcal{C}$, such that the price of anarchy of atomic flow routing in $(G,R,\Phi)$ is at least $\frac{(a_{max})^{1/2}+2}{3}$, i.e., $\Omega(\sqrt{a_{max}})$.%\todo{$a$ is ??}
\end{lemma}

We also show an almost matching upper bound in the following.

%\subsubsection{Almost tight upper bound for Uniform Functions}
%\subsection{Improved PoA for Uniform Two-Commodity Networks, Affine Delays}
% \subsubsection{Lower Bounds on PoA for Uniform Delay Functions}
% The proof of Lemma \ref{lower-bound-affine-uniform-atomic} is shown below.
% For the class of affine functions of the form $af_1 + f_2$ with nonnegative constant, we provide a simple example which shows PoA is $\Omega(\sqrt{a})$.
%   \begin{figure}[ht]
%   \begin{center}
%   \scalebox{0.55}
%   {\includegraphics[width=155mm ,height=30mm]{atomic-uniform.eps}}
%   \caption{Uniform and affine latency}
%   \label{fig:comm-dependent-poa}
%   \end{center}
%   \end{figure}
% In this network, the
% top edge $e$ is associated with the delay function
% $\Phi_{e}(f_1(e),f_2(e)) = af_1 + f_2 + a$ and the bottom edge $h$ is associated with the delay function
% $\Phi_{h}(f_1(h),f_2(h)) = af_1 + f_2$.
% Demand requirements are defined as $r_1 = 1$ and $r_2 = (a)^{1/2}$.
% The Nash equilibrium(NE) flow vector $f$ has the flow $f_1(h) = 1$ and $f_2(e) = (a)^{1/2}$ and consequently the
% cost of NE $C_{NE}(f) = a((a)^{1/2}+2)$. Conversely the social optimum flow $\hat{f}$
% can be obtained from $\hat{f}_1(e) = 1$ and $\hat{f}_2(h) = (a)^{1/2}$
% and the social optimum cost $C_{SO}(\hat{f}) = 3a$.
% The price of anarchy is thus $\frac{a((a)^{1/2}+2)}{3a} = \frac{(a)^{1/2}+2}{3}$.

\subsubsection{Almost tight upper bound for Uniform Functions}
\begin{theorem}\label{thm-2-atomic-uniform-affine}
Let $\mathcal{C}$ be a set of  affine delay functions that are uniform up to differences in the constant term.
If $(G,R,\Phi)$ is an atomic 2-commodity network flow routing game instance with delay functions in $\mathcal{C}$, then the price of anarchy of atomic flow routing in $(G,R,\Phi)$
is at most $\sqrt{a_{max}} + 2$ and this is  almost tight as shown in \Cref{lower-bound-affine-uniform-atomic}.
\end{theorem}

The proof of \Cref{thm-2-atomic-uniform-affine} is shown below.  Recall that the theorem claims that if $(G,R,\Phi)$ is an atomic 2-commodity instance with uniform affine delay functions, then the price of anarchy of $(G,R,\Phi)$
is at most $\sqrt{a_{max}} + 2$.

\paragraph{Proof Outline:} The proof proceeds by comparing both the Nash equilibrium solution and the optimum solution. Removing the common flows between the two, a set of cycles are obtained such that by reversing flow on these cycles, one flow can be transformed to the other. The price of anarchy is obtained using the structure of these cycles. We show that we can obtain the PoA by considering one cycle only and analyzing the worst case behavior depending on how the two, the Nash equilibrium and the optimum flows, utilize the edges of the cycle. The analysis is by cases; we obtain  sixteen possibilities but three cases dominate all other cases and thus we provide an upper bound of the PoA for these three cases. The analysis uses variational inequalities to determine the relationship between the Nash equilibrium and optimum flow.

\paragraph{Reducing to Cycles:}
Consider $F_{\bigotimes} = (F - \hat{F}) \bigcup (\hat{F} - F)$ where $F$ and $\hat{F}$ are a Nash equilibrium flow and a social optimum flow, respectively. Further we let $F_1$ and $F_2$ denote the sub-flows of $F$ at a Nash equilibrium flow for commodity 1 and commodity 2, respectively. Similarly, we let $\hat{F}_1$ and $\hat{F}_2$ denote the social optimum flow for commodity 1 and commodity 2, respectively.
Note that $\bigotimes$ cancels common flows in $F$ and $\hat{F}$.
In the figures below that illustrate the various cases, we let edges in $F_1 \bigotimes \hat{F}_1$ be thicker and edges in $F_2 \bigotimes \hat{F}_2$ be thinner lines.

By reversing the flow in $\hat{F}$ and partitioning $F_1 \bigotimes \hat{F}_1$ and $F_2 \bigotimes \hat{F}_2$, respectively, we obtain a set of cycles for commodity 1 and commodity 2.
We let the cycles represented by thicker lines be denoted by $CC^1$ for commodity 1 and cycles represented by the
thinner lines by $CC^2$ for commodity 2.
Note that an edge could be utilized by commodity 1 as well as commodity 2, and it can occur in cycles of both commodities.

\paragraph{Analyzing PoA bounds for Cycles:}
We need a lemma which show that there exists a good partition of the cycles so that we can obtain an almost tight upper bound.
Let $(a_i,b_i)$ be a collection of pairs where $1 \leq i \leq d: a_i,b_i \in \mathbb{R}_+$.
Let $\Pi$ be a set of partitions of $I = \{1,\ldots,d\}$ such that
\begin{itemize}
\item If $\pi=(\pi^1,\ldots,\pi^{\ell}) \in \Pi$ then $|\pi^j| \leq 2$,

Note the following are true

\item If $\pi=(\pi^1,\ldots,\pi^{\ell}) \in \Pi$ then $\bigcup_j\pi^j = I$,
\item If $\pi^{j_1}, \pi^{j_2} \in \pi$ and $\pi^{j_1} \neq \pi^{j_2}$ then $\pi^{j_1} \bigcap \pi^{j_2} = \emptyset$.
\end{itemize}
\begin{lemma}\label{2cycle-enough}
Let $\hat{\pi} \in \Pi$ be an optimum partition of $I$ defined as
\[
\hat{\pi}= \arg \min_{\pi \in \Pi } \max_{\pi^j \in \pi}\frac{\sum_{\ell \in \pi}a_{\ell}}{\sum_{\ell \in \pi}b_{\ell}}
\]
% or
% \[ \forall \pi \in \Pi : \max_{\pi^{j'} \in \pi}\frac{\sum_{\ell \in \pi^{j'}}a_{\ell}}{\sum_{\ell =\in \pi^{j'}}b_{\ell}} \geq \max_{\hat{\pi}^j \in \hat{\pi}}\frac{\sum_{\ell \in \hat{\pi}}a_{\ell}}{\sum_{\ell \in \hat{\pi}}b_{\ell}} \]
% where $\hat{\pi} = (\hat{\pi}^1,\ldots,\hat{\pi}^q)$ and $\hat{\pi}^j$ corresponds to the $j$-th part and is a subset (of at most 2 elements) of $I$ of size $t_j$, i.e. $\hat{\pi}^j = \{x_1, \ldots, x_{t_j}\}$. 
where $\pi = (\pi^1,\ldots,\pi^q)$ and $\pi^j$ corresponds to the $j$-th part and is a subset (of at most 2 elements) of $I$ of size $t_j$, i.e., $\pi^j = \{x_1, \ldots, x_{t_j}\}$. 
Then
\[ \frac{a_1+a_2+\ldots+a_d}{b_1+b_2+\ldots+b_d} \leq \max_{\hat{\pi}^j \in \hat{\pi}}\frac{\sum_{\ell \in \hat{\pi}^j}a_{\ell}}{\sum_{\ell \in \hat{\pi}^j}b_{\ell}}. \]
\end{lemma}

We next show that PoA can be estimated by considering at most one  cycle per commodity instead of all cycles.
Throughout this subsection, let $a = a_1/a_2$ where $a_1$ and $a_2$ represent coefficients of commodity 1 and commodity 2 in the affine uniform delay function.

Let $\mathcal{CC}$ be the set of all cycles and $\Pi$ be a set of partitions of $\mathcal{CC}$.
Consider a partition $\bar{\pi}$ of $\mathcal{CC}$ where each part is of size either 1 or of size 2 (in which case it contains 1 cycle from each commodity, $C^1 \in CC^1$ and $C^2 \in CC^2$).
Let the pair $(f_1, f_2)$ and $(\hat{f}_1, \hat{f}_2)$ be the Nash equilibrium flow of commodity 1 and commodity 2 and the social optimum flow of commodity 1 and commodity 2, respectively.
Let $\hat{\pi} = \{\hat{\pi}^1, \ldots, \hat{\pi}^q\}$ be a partition which guarantees the optimum amongst partitions in $\Pi$. 
Then,
\begin{eqnarray*}
\phantom{a} \frac{C_{NE}(f)}{C_{SO}(\hat{f})} &=& \frac{\sum_{e \in E}(f_1(e)+f_2(e))\Phi(f_1(e),f_2(e))}{\sum_{e \in E}(\hat{f}_1(e)+\hat{f}_2(e))\Phi(\hat{f}_1(e),\hat{f}_2(e))} \label{one-cycle-consider-1} \\
\phantom{a}  &\leq& \max_{\hat{\pi}^j \in \hat{\pi}}\left( \frac{\textit{cost of NE flow in cycles or combinations of cycles in } \hat{\pi}^j}{\textit{cost of SO flow in cycles or combinations of cycles in } \hat{\pi}^j} \right)\\
\phantom{a}  &\leq& \max_{\bar{\pi}^j \in \bar{\pi}}\left( \frac{\textit{cost of NE flow in cycles or combinations of cycles in } \bar{\pi}^j}{\textit{cost of SO flow in cycles or combinations of cycles in } \bar{\pi}^j} \right)\\
\phantom{a}  &\leq& \max \left \{ \right. \\
& & \left. \max_{\bar{\pi}^j = (C^1,C^2)} \frac{f_1\Phi_{C^1}(f_1,f_2)+f_2\Phi_{C^2}(f_1,f_2)}{\hat{f}_1\Phi_{C^1}(\hat{f}_1,\hat{f}_2)+\hat{f}_2\Phi_{C^2}(\hat{f}_1,\hat{f}_2)}, \right. \\
& & \left. \max_{\bar{\pi}^j = (C^1)}\frac{f_1\Phi_{C^1}(f_1,f_2)}{\hat{f}_1\Phi_{C^1}(\hat{f}_1,\hat{f}_2)}, \right. \\
& & \left. \max_{\bar{\pi}^j = (C^2)}\frac{f_2\Phi_{C^2}(f_1,f_2)}{\hat{f}_2\Phi_{C^2}(\hat{f}_1,\hat{f}_2)}  \right. \\
& & \left. \right \}
\end{eqnarray*}
where $\Phi_{C^1}()$ and $\Phi_{C^2}()$ represent cost incurred in $C^1$ and $C^2$, respectively.
The first inequality holds due to \Cref{2cycle-enough}.
In the last inequality, we have $|\bar{\pi}^j| = 2$ for the first factor and $|\bar{\pi}^j| = 1$ for the second and third factors.
In \Cref{2cycle-enough}, we showed that the maximum cost element in $\hat{\pi}$ provides an upper bound for $C_{NE}(f)/C_{SO}(\hat{f})$.
So, the partition $\bar{\pi}$ considered here provides an upper bound for $C_{NE}(f)/C_{SO}(\hat{f})$.
Later we will show that this partition is good enough to obtain an almost tight bound for $C_{NE}(f)/C_{SO}(\hat{f})$.

To analyze the price of anarchy, we consider the structure of the cycles.

\paragraph{Cycle Structure:}
In this subsection we consider the structure of the cycles obtained in $F \bigotimes \hat{F}$.
One type of  cycle obtained is illustrated in \Cref{fig:intersection-graph}.
Here $u$ and $v$ are starting and end nodes such that reversing of $\hat{f}_1$ leads to a cycle.
Though there are cycles for commodity 1 and commodity 2, we consider cycles for commodity 1.
We consider flow from commodity 1 as a primary flow and flow from commodity 2 as a secondary flow, respectively.
The case when flow from commodity 2 is a primary flow is symmetric to this case.
Note that though we consider cycles from commodity 1, it is possible that those might be intersected with cycles from commodity 2.
We show an example of a cycle from commodity 1 intersecting with a flow of commodity 2 in \Cref{fig:intersection-graph}.
\begin{figure}[ht]
  \begin{center}
  \scalebox{0.7}
  {\includegraphics[width=35mm ,height=30mm]{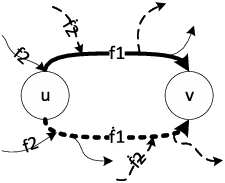}}
  \caption{An example of a cycle from commodity 1}
  \label{fig:intersection-graph}
  \end{center}
\end{figure}
In the figure, % \ref{fig:intersection-graph}, 
the thinner and solid line represent $f_2$; while the thinner and dashed lines represent $\hat{f}_2$.
As shown in \Cref{fig:intersection-graph}, $f_2$ and $\hat{f}_2$  intersect with a cycle of commodity 1 and also $f_2$ and $\hat{f}_2$ can  overlap each other.

Depending on how the NE flow and social optimum flow use the top and bottom path in the cycle, we have  sixteen possibilities.
We show that three cases dominate all other cases and thus we provide an upper bound of the PoA for these three cases.

{\em Notation:}
As shown in \Cref{fig:intersection-graph}, we have two paths, termed as top path and bottom path.
Throughout this section, let $\Phi_{P_i}(f_1,f_2)$ be cost of path $P_i$ when flow on path $P$ of commodity 1 is $f_1$ and that on path $P$ of commodity 2 is $f_2$.
We define  $\Phi'_{P_i}(f_1,f_2)$ as the cost when restricted to  a subset of edges $P' \subseteq P_i$, which is utilized by both $f_1$ and $f_2$.
Similarly for $P'' \subseteq P$ let $\Phi''_{P_i}(f_1,0)$ (or $\Phi''_{P_i}(0,f_2)$) be defined as the cost function over a set of edges $P''$ which is utilized by $f_1$ (or $f_2$), but not both.
Note that $P' \bigcap P'' = \emptyset$.
For simplicity, we define $\Phi_{P_i}(f_1,f_2) = \Phi_i(f_1,f_2)$, $\Phi'_{P_i}(f_1,f_2) = \Phi'_i(f_1,f_2)$, $\Phi''_{P_i}(f_1,0) = \Phi''_i(f_1,0)$ and $\Phi''_{P_i}(0,f_2) = \Phi''_i(0,f_2)$.

Also, throughout this subsection we let $P_1, P_2$ correspond to edges containing flows of both commodities in the top path and the bottom path from $u_1$ to $v_1$, respectively. And we let $ P_3$ and $P_4$ correspond to the edges common to both flows on the top path and the  bottom path from $u_2$ to $v_2$, respectively as shown in \Cref{fig:intersection-graph1}.
Lastly, we define $\ell_1 = |\{e \in P_1 \bigcap E\}|, \ell_2 = |\{e \in P_2 \bigcap E\}|, \ell_3 = |\{e \in P_3 \bigcap E\}|$ and $\ell_4 = |\{e \in P_4 \bigcap E\}|$.
\begin{figure}[ht]
  \begin{center}
  \scalebox{0.7}
  {\includegraphics[width=75mm ,height=32mm]{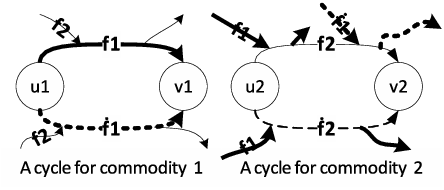}}
  \caption{Two cycles - one cycle per each commodity.}
  \label{fig:intersection-graph1}
  \end{center}
\end{figure}
Remember that $f_1 = \hat{f}_1 = r_1$ and $f_2 = \hat{f}_2 = r_2$ due to the definition of atomic network model in this section.

\paragraph{Three cases are enough to be considered:}
Note the $f$ and $\hat{f}$ correspond to a NE flow and a social optimum flow.
\begin{figure}[ht]
  \begin{center}
  \scalebox{0.8}
  {\includegraphics[width=190mm ,height=40mm]{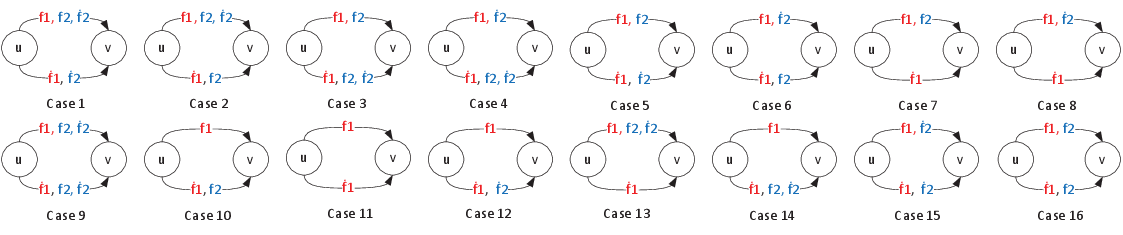}}
  \caption{All Cases}
  \label{fig:all-cases}
  \end{center}
\end{figure}
We consider all possible cases when commodity 1 flow is considered as a primary flow as shown in \Cref{fig:all-cases}. (The notation used is $S$ vs $S'$, where $S$ is the set of flows on the top path and $S'$ the set of flows on the bottom path): (1) $\{f_1,f_2,\hat{f}_2 \}$ vs.
 $\{ \hat{f}_1,\hat{f}_2 \}$, (2) $\{f_1,f_2,\hat{f}_2 \}$ vs. $\{ \hat{f}_1,f_2 \}$,
 (3) $\{f_1,f_2\}$ vs. $\{\hat{f}_1,f_2,\hat{f}_2\}$, (4) $\{f_1,\hat{f}_2\}$ vs. $\{ \hat{f}_1,f_2,\hat{f}_2\}$,
 (5) $\{f_1,f_2\}$ vs. $\{\hat{f}_1,\hat{f}_2\}$, (6) $\{f_1,\hat{f}_2\}$ vs. $\{\hat{f}_1,f_2\}$, (7) $\{f_1,f_2\}$ vs. $\{\hat{f}_1\}$,
 (8) \{$f_1,\hat{f}_2$\} vs. \{$\hat{f}_1$\}, (9) \{$f_1,f_2,\hat{f}_2$\} vs. \{$\hat{f}_1,f_2,\hat{f}_2$\}, (10) \{$f_1$\} vs. \{$\hat{f}_1,f_2$\}, (11) \{$f_1$\} vs. \{$\hat{f}_1\}$, (12) \{$f_1$\} vs. \{$\hat{f}_1,\hat{f}_2$\}, (13) \{$f_1,f_2,\hat{f}_2$\} vs. \{$\hat{f}_1$\}, (14) \{$f_1$\} vs. \{$\hat{f}_1,f_2,\hat{f}_2$\}, (15) \{$f_1,\hat{f}_2$\} vs. \{$\hat{f}_1,\hat{f}_2\}$ and (16) \{ $f_1,f_2$\} vs. \{$\hat{f}_1,f_2$\}.
By reversing $\hat{f}_1$, each of the structures considered becomes a directed cycle.
As mentioned before $f_2$ or $\hat{f}_2$ (or both) use edges in  these cycles.
\begin{table}[ht]
\begin{tabular}{|l|c|c|}
\hline
\mbox{} & ratio of $C_{NE}(f)$ to $C_{SO}(\hat{f})$ & variational inequality \\
\hline\hline
case 1,5 & $\frac{f_1\Phi_1(f_1,f_2)}{\hat{f}_1\Phi'_2(\hat{f}_1,\hat{f}_2)+\hat{f}_1\Phi''_2(\hat{f}_1,0)}$ & $\Phi_1(f_1,f_2) \leq \Phi_2(f_1,0),\Phi_1(f_1,f_2) \leq \Phi_2(0,f_2)$ \\ \hline
case 2,16 & $\frac{f_1\Phi_1(f_1,f_2)}{\hat{f}_1\Phi_2(\hat{f}_1,0)}$ & $\Phi_1(f_1,f_2) \leq \Phi'_2(f_1,f_2)+\Phi''_2(f_1,0)$ \\ \hline
case 3,9 & $\frac{f_1\Phi_1(f_1,f_2)}{\hat{f}_1\Phi'_2(\hat{f}_1,\hat{f}_2)+\hat{f}_1\Phi''_2(\hat{f}_1,0)}$ & $\Phi_1(f_1,f_2) \leq \Phi'_2(f_1,f_2)+\Phi''_2(f_1,0)$ \\ \hline
case 4,14 & $\frac{f_1\Phi_1(f_1,0)}{\hat{f}_1\Phi'_2(\hat{f}_1,\hat{f}_2)+\hat{f}_1\Phi''_2(\hat{f}_1,0)}$ & $\Phi_1(f_1,0) \leq \Phi'_2(f_1,f_2)+\Phi''_2(f_1,0)$ \\ \hline
case 6,10,11 & $\frac{f_1\Phi_1(f_1,0)}{\hat{f}_1\Phi_2(\hat{f}_1,0)}$ & $\Phi_1(f_1,0) \leq \Phi'_2(f_1,f_2)+\Phi''_2(f_1,0)$ \\ \hline
case 7,13 & $\frac{f_1\Phi_1(f_1,f_2)}{\hat{f}_1\Phi_2(\hat{f}_1,0)}$ & $\Phi_1(f_1,f_2) \leq \Phi_2(f_1,0), \Phi_1(f_1,f_2) \leq \Phi_2(0,f_2)$ \\ \hline
case 8 & $\frac{f_1\Phi_1(f_1,0)}{\hat{f}_1\Phi_2(\hat{f}_1,0)}$ & $\Phi_1(f_1,0) \leq \Phi_2(f_1,0)$ \\ \hline
case 12,15 & $\frac{f_1\Phi_1(f_1,0)}{\hat{f}_1\Phi'_2(\hat{f}_1,\hat{f}_2)+\hat{f}_1\Phi''_2(\hat{f}_1,0)}$ & $\Phi_1(f_1,0) \leq \Phi_2(f_1,0)$ \\ \hline
\end{tabular}
\caption{$C_{NE}(f)/C_{SO}(\hat{f})$ for all cases}
\label{table:16cases}
\end{table}
\Cref{table:16cases}  lists the cases and a corresponding (variational) inequality (we use the term, {\em variational inequality} here imprecisely simply to emphasize properties of the solution ) obtained by the flow being Nash equilibrium.
%\todo{These are just inequalities-not variational inequalities}
Let us consider case 16 to show an example of how to construct formulas for $C_{NE}(f)/C_{SO}(\hat{f})$ using the corresponding variational inequality.
There are two NE flows $f_1$ and $f_2$ over the top path, and the cost of NE for commodity 1 is $f_1\Phi_1(f_1,f_2)$; while, the cost of SO for commodity 1 is incurred only by $\hat{f}_1$ which results in $\hat{f}_1\Phi_2(\hat{f}_1,0)$.
Its corresponding variational inequality is $\Phi_1(f_1,f_2) \leq \Phi'_2(f_1,f_2)+\Phi''_2(f_1,0)$ since a shift of flow $f_1$ to the bottom path (note that $f_2$ already exists) does not decrease the cost incurred by $f_1$ and $f_2$ over the top path.

\paragraph{(i) cases 1, 5, 7 and 13:}\label{subsubsection-15713}
Note that both case 1 and case 5 have the same ratio  of $C_{NE}(f)$ to $C_{SO}(\hat{f})$ and variational inequality, and case 7 and case 13 use the same ratio of $C_{NE}(f)$ to $C_{SO}(\hat{f})$ and variational inequality.
Then price of anarchy as computed in  case 1 and case 5 is upper bounded by that computed in case 7(case 13) since the divisor in case 7(case 13) is smaller than divisors in other two cases since $f_1 = \hat{f}_1 = r_1$ and $f_2 = \hat{f}_2 = r_2$.

In case 7 (or 13), note that there is no social optimum flow from commodity 2 in the ratio of $C_{NE}(f)$ to $C_{SO}(\hat{f})$, and the NE cost incurred by commodity 2 can be considered from cycles with flow of commodity 2. 
However, $f_2$ has impact on latency of $f_1$, and we consider $f_1\Phi_1(f_1,f_2)$ as cost incurred by $f_1$ and $f_2$.
\[ \phantom{a} \frac{C_{NE}(f)}{C_{SO}(\hat{f})} \leq  \frac{f_1\Phi_1(f_1,f_2)}{\hat{f}_1\Phi_2(\hat{f}_1,0)} \leq \frac{r_1\Phi_1(r_1,r_2)}{r_1\Phi_2(r_1,0)} \leq \frac{r_1\Phi_1(r_1,r_2)}{r_1\Phi_1(r_1,r_2)} = 1.\]
The first inequality holds  since $f_1 = \hat{f}_1 = r_1$ and $f_2 = \hat{f}_2 = r_2$ due to the definition of atomic unsplittable flow.
The last inequality holds since $\Phi_1(f_1,f_2) \leq \Phi_2(f_1,0)$ due to its corresponding variational inequality.

\paragraph{(ii) cases 2, 3, 9 and 16:}\label{subsubsection-23916}
Since the divisor for case 2 and case 16 is smaller than divisors in case 3 and case 9, the price of anarchy in case 3 and case 9 are dominated by the other two cases.
For both case 2 or case 16, we have the following inequality:
\begin{eqnarray}
\phantom{a} \frac{C_{NE}(f)}{C_{SO}(\hat{f})} &\leq& \frac{f_1\Phi_1(f_1,f_2)}{\hat{f}_1\Phi_2(\hat{f}_1,0)} \label{case-239-1}\\
\phantom{a} &\leq& \frac{f_1\Phi'_2(f_1,f_2)+f_1\Phi''_2(f_1,0)}{\hat{f}_1\Phi_2(\hat{f}_1,0)} \label{case-239-2}\\
\phantom{a} &=& \frac{r_1\Phi'_2(r_1,r_2)+r_1\Phi''_2(r_1,0)}{r_1\Phi_2(r_1,0)} \label{case-239-3}\\
\phantom{a} &\leq& \frac{\Phi_2(r_1,r_2)}{\Phi_2(r_1,0)} \label{case-239-4}\\
\phantom{a} &=& \frac{a_2\ell_2(ar_1+r_2+c'_2)}{a_2\ell_2(ar_1+c'_2)} \leq 1 + \frac{r_2}{ar_1} \label{case-239-5}
\end{eqnarray}
In (\ref{case-239-2}),  $\Phi_1(f_1,f_2) \leq \Phi'_2(f_1,f_2)+\Phi''_2(f_1,0)$ due to the corresponding variational inequality.
We have $c'_2 = c_2/a_2$ in (\ref{case-239-5}).
It seems that the price of anarchy for these cases are not bounded since $f_2$ appears in the dividend while the divisor does not include $\hat{f}_2$.
However, if $r_2$ is big enough (larger than $a\sqrt{a}r_1$) then $r_2$ play an important role.
Thus, we consider splitting into two cases : $r_2 \leq a\sqrt{a}r_1$ and $r_2 > a\sqrt{a}r_1$.
In the case that $r_2 \leq a\sqrt{a}r_1$, the price of anarchy is bounded by $\sqrt{a}+1$.

We next discuss the price of anarchy when $r_2 > a\sqrt{a}r_1$.
Note that we assume that there is an intersection of a commodity 2 cycle with this cycle when $r_2 > a\sqrt{a}r_1$.
If there is no intersection with a  commodity 2 cycle, then $f_1$ can be shifted from the top path to the bottom path leading to less delay.
If the delay on the top path and the delay on the bottom path are the same, then there is no cycle.
Suppose that the shift from the top path to the bottom path of $f_1$ does not decrease delay.
It implies that the cost of the bottom path is greater than the cost of the top path and potentially violates utilizing the bottom path by $\hat{f}_1$ is less cost than using the top path.

Let us consider the cycle of commodity 2.
$f_2$ goes along top path overlapped with $f_1$ or $\hat{f}_1$ on the bottom path.
The flow $f_2$ and the other  flow $\hat{f}_2$ forms the cycle of commodity 2 which is intersecting with a cycle shown in case 2.
For this cycle, we derive the price of anarchy via the  corresponding variational inequality.
\begin{eqnarray}
\phantom{a} \frac{C_{NE}(f)}{C_{SO}(\hat{f})} &\leq& \frac{f_2\Phi'_1(f_1,f_2)+f_2\Phi_2(0,f_2)+f_2\Phi'_3(f_1,f_2)+f_2\Phi''_3(0,f_2)}{\hat{f}_1\Phi_2(\hat{f}_1,0)+\hat{f}_2\Phi_4(0,\hat{f}_2)} + \nonumber \\
\phantom{a} & & \frac{f_1\Phi'_1(f_1,f_2)+f_1\Phi''_1(f_1,0)}{\hat{f}_1\Phi_2(\hat{f}_1,0)+\hat{f}_2\Phi_4(0,\hat{f}_2)}\\
\phantom{a} &\leq& \frac{f_1\Phi'_2(f_1,f_2)+f_1\Phi''_2(f_1,0)+f_2\Phi_4(f_1,f_2)}{\hat{f}_1\Phi_2(\hat{f}_1,0)+\hat{f}_2\Phi_4(0,\hat{f}_2)} \label{case-239-7}\\
\phantom{a} &\leq& \frac{f_1\Phi_2(f_1,f_2)+f_2\Phi_4(f_1,f_2)}{\hat{f}_1\Phi_2(\hat{f}_1,0)+\hat{f}_2\Phi_4(0,\hat{f}_2)} \label{case-239-8}\\
\phantom{a} &=& \frac{r_1\Phi_2(r_1,r_2)+r_2\Phi_4(r_1,r_2)}{r_1\Phi_2(r_1,0)+r_2\Phi_4(0,r_2)} \label{case-239-9}\\
\phantom{a} &=& \frac{r_1\ell_2(ar_1+r_2+c'_2)+r_2\ell_4(ar_1+r_2+c'_4)}{r_1\ell_2(ar_1+c'_2)+r_2\ell_4(r_2+c'_4)} \label{case-239-10}\\
\phantom{a} &=& \frac{r_1\ell(ar_1+r_2+c'_2)+r_2(ar_1+r_2+c'_4)}{r_1\ell(ar_1+c'_2)+r_2(r_2+c'_4)} \label{case-239-11}\\
\phantom{a} &=& 1 + \frac{r_1r_2\ell+ar_1r_2}{r_1\ell(ar_1+c'_2)+r_2(r_2+c'_4)}. \label{case-239-12}
\end{eqnarray}
We have $\Phi'_1(f_1,f_2)+\Phi_2(0,f_2)+\Phi'_3(f_1,f_2)+\Phi''_3(0,f_2) \leq \Phi_4(f_1,f_2)$ due to the variational inequality for commodity 2 and $\Phi'_1(f_1,f_2)+\Phi''_1(f_1,0) \leq \Phi'_2(f_1,f_2)+\Phi''_2(f_1,0)$ due to  the variational inequality for commodity 1.
By these observations, we can obtain (\ref{case-239-7}) from the previous equation.
Inequality (\ref{case-239-8}) can be obtained by summing over $\Phi'_2(f_1,f_2)$ and $\Phi''_2(f_1,0)$ by replacing $\Phi''_2(f_1,0)$ as $\Phi''_2(f_1,f_2)$.
Let $\ell = \ell_2/\ell_4$ where $\ell_2$ and $\ell_4$ represent the number of edges on path 2 and path 4, respectively.
Also, let $c'_2 = c_2/a_2, c'_4 = c_4/a_2$ throughout this chapter.
From the above variational inequality $\Phi'_1(f_1,f_2)+\Phi_2(0,f_2)+\Phi'_3(f_1,f_2)+\Phi''_3(0,f_2) \leq \Phi_4(f_1,f_2)$ we have $\Phi_2(0,f_2) \leq \Phi_4(f_1,f_2)$ and further
\[ \ell = \ell_2/\ell_4 \leq \frac{ar_1+r_2+c'_4}{r_2+c'_2} \leq \frac{ar_1+r_2+c'_4}{r_2}.\]
When $\frac{ar_1+r_2+c'_4}{r_2+c'_2} \leq 1$, let $\ell$ be 0 in the divisor and be 1 for the dividend. Then, equation (\ref{case-239-12}) can be written as
\begin{eqnarray}
\phantom{a} \frac{C_{NE}(f)}{C_{SO}(\hat{f})} \leq 1 + \frac{r_1r_2+ar_1r_2}{r_2(r_2+c'_4)} \leq 1 + \frac{r_1r_2+ar_1r_2}{r^2_2} \leq 2 + \frac{1}{\sqrt{a}}.
\end{eqnarray}
The last inequality holds because $r_2 \geq a\sqrt{a}r_1$, and $a \geq 1$.
When $\frac{ar_1+r_2+c'_4}{r_2+c'_2} \geq 1$, note that $\frac{ar_1+r_2+c'_4}{r_2} \geq 1$ is true.
Equation (\ref{case-239-12}) can be written as
\begin{eqnarray}
\phantom{a} \frac{C_{NE}(f)}{C_{SO}(\hat{f})} &\leq& 1 + \frac{r_1r_2(ar_1+r_2+c'_4)/r_2+ar_1r_2}{r_1\ell(ar_1+c'_2)+r_2(r_2+c'_4)} \label{case-239-13} \\
\phantom{a} &\leq& 1 + \frac{r_1(r_2/\sqrt{a}+r_2+c'_4)+ar_1r_2}{r_1\ell(ar_1+c'_2)+a\sqrt{a}r_1(r_2+c'_4)} \label{case-239-14} \\
\phantom{a} &\leq& 1 + \frac{(a+1+1/\sqrt{a})r_2+c'_4}{\ell(ar_1+c'_2)+a\sqrt{a}(r_2+c'_4)} \label{case-239-15} \\
\phantom{a} &\leq& 1 + \frac{(a+1+1/\sqrt{a})r_2+c'_4}{a\sqrt{a}(r_2+c'_4)} \leq 1 + \frac{1}{\sqrt{a}} + \frac{2}{a\sqrt{a}}\label{case-239-16}
\end{eqnarray}
Observe that $\ell(ar_1+c'_2)$ can be ignored to obtain the upper bound in (\ref{case-239-16}).
In equations, to maximize the price of anarchy we substitute $r_1$ with $r_2/a\sqrt{a}$ in the dividend; while in the divisor we substitute $r_2$ with $a\sqrt{a}r_1$.

\paragraph{(iii) cases 4, 6, 8, 10, 11, 12, 14 and 15:}
Note that the divisors in case 6, 10 and 11 are smaller than the divisors in other cases, and the dividend in case 6, 10 and 11 are bigger than others due to the variational inequality.

We consider one of the cases 6, 10 and 11.
As shown in previous case, we need to consider commodity 1's cycle and commodity 2's cycle to estimate the price of anarchy.
We provide a proof of a case when $\hat{f}_2$ is utilizing the top path.
Let us consider case 6 which upper bounds other two cases.
we split into three sub-cases: (a) $r_1 \leq r_2 \leq (a)^{1/2}r_1$, (b) $r_1 > r_2$ and (c) $r_2 > (a)^{1/2}r_1$.
In case (a),
\begin{eqnarray}
\phantom{a} \frac{C_{NE}(f)}{C_{SO}(\hat{f})} &\leq& \frac{f_1\Phi_1(f_1,0)+f_2\Phi_2(0,f_2)}{\hat{f}_1\Phi_2(\hat{f}_1,0)+\hat{f}_2\Phi_1(0,\hat{f}_2)} \label{atomic-affine-overlap-2}\\
\phantom{a}  &\leq& \frac{f_1\Phi_2(f_1,f_2)+f_2\Phi_2(0,f_2)}{\hat{f}_1\Phi_2(\hat{f}_1,0)+\hat{f}_2\Phi_1(0,\hat{f}_2)}  \label{atomic-affine-overlap-3}\\
\phantom{a}  &\leq& \frac{r_1\Phi_2(r_1,r_2)+r_2\Phi_2(0,r_2)}{r_1\Phi_2(r_1,0)+r_2\Phi_1(0,r_2)} \label{atomic-affine-overlap-5} \\
\phantom{a}  &\leq& \frac{r_1\Phi_2(r_1,r_2)+r_2\Phi_2(0,r_2)}{r_1\Phi_2(r_1,0)} \label{atomic-affine-overlap-6}\\
\phantom{a}  &\leq& \frac{\Phi_2(r_1,r_2)}{\Phi_2(r_1,0)}+\frac{r_2}{r_1}\frac{\Phi_2(0,r_2)}{\Phi_2(r_1,0)} \label{atomic-affine-overlap-7}\\
\phantom{a}  &\leq& 1 + \frac{\ell_2\sqrt{a}r_1}{\ell_2(ar_1+c'_2)}+\frac{r_2}{r_1}\frac{\ell_2(\sqrt{a}r_1+c'_2)}{\ell_2(ar_1+c'_2)} \label{atomic-affine-overlap-8}\\
\phantom{a}  &\leq& 1/\sqrt{a} + 2
\end{eqnarray}
Due to the variational inequality, we have (\ref{atomic-affine-overlap-3}) by replacing $\Phi_1(f_1,0)$ with $\Phi_2(f_1,f_2)$.
In inequality (\ref{atomic-affine-overlap-8}), we substitute $r_2$ with $\sqrt{a}r_1$.
In case (b), we start at (\ref{atomic-affine-overlap-5}) to avoid duplicate formulas.
\begin{eqnarray}
\phantom{a} \frac{C_{NE}(f)}{C_{SO}(\hat{f})} &\leq& \frac{r_1\ell_1(ar_1+r_2+c'_2)+r_2\ell_2(r_2+c'_2)}{r_1\ell_1(ar_1+c'_2)+r_2\ell_2(r_2+c'_1)} \label{atomic-affine-overlap-2-1}\\
\phantom{a} &\leq& 1 + \frac{\ell_1r_1r_2+\ell_2r_2c'_2}{\ell_1(ar^2_1+r_1c'_2)+\ell_2(r^2_2+r_2c'_1)} \label{atomic-affine-overlap-2-2}\\
\phantom{a} &\leq& 1 + \max(\frac{r_1r_2}{ar^2_1+r^2_2},\frac{r_2c'_2}{r_1c'_2+r_2c'_1}) \label{atomic-affine-overlap-2-3} \\
\phantom{a} &\leq& 1 + \max(\frac{1}{ar_1/r_2+r_2/r_1},1) \label{atomic-affine-overlap-2-4} \leq 1
\end{eqnarray}
Due to \Cref{2cycle-enough}, (\ref{atomic-affine-overlap-2-3}) holds.
In equation (\ref{atomic-affine-overlap-2-4}), the divisor is minimized when $r_1/r_2$ is close to $1$.
Lastly, in case (c),
\begin{eqnarray}
\phantom{a} \frac{C_{NE}(f)}{C_{SO}(\hat{f})} &\leq& \frac{\ell_1f_1(af_1+c'_1)+\ell_2f_2(f_2+c'_2)}{\ell_1\hat{f}_2(\hat{f}_2+c'_1)+\ell_2\hat{f}_1(a\hat{f}_1+c'_2)} \label{atomic-affine-overlap-3-1}\\
\phantom{a} &\leq& \frac{\ell f_1(af_1+c'_1)+f_2(f_2+c'_2)}{\ell \hat{f}_2(\hat{f}_2+c'_1)+\hat{f}_1(a\hat{f}_1+c'_2)} \label{atomic-affine-overlap-3-2}\\
\phantom{a} &\leq& \frac{\ell r_1(ar_1+c'_1)+r_2(r_2+c'_2)}{\ell r_2(r_2+c'_1)+r_1(ar_1+c'_2)} \label{atomic-affine-overlap-3-3}\\
\phantom{a} &\leq& \frac{r_1(ar_1+c'_1)\frac{ar_1+r_2+c'_2}{ar_1+c'_1}+r_2(r_2+c'_2)}{r_2(r_2+c'_1)\frac{r_2+c'_2}{ar_1+r_2+c'_1}+r_1(ar_1+c'_2)} \label{atomic-affine-overlap-3-4}\\
\phantom{a} &\leq& \frac{r_1(ar_1+r_2+c'_2)+r_2(r_2+c'_2)}{r_2\frac{r_2+c'_2}{\sqrt{a}+1}+r_1(ar_1+c'_2)} \label{atomic-affine-overlap-3-5}\\
\phantom{a} &\leq& \frac{(r_1r_2}{r_2\frac{r_2+c'_2}{\sqrt{a}+1}+r_1(ar_1+c'_2)} + \frac{r_1(ar_1+c'_2)+r_2(r_2+c'_2)}{r_2\frac{r_2+c'_2}{\sqrt{a}+1}+r_1(ar_1+c'_2)} \label{atomic-affine-overlap-3-6}\\
\phantom{a} &\leq& \frac{r_1r_2}{\frac{r^2_2}{\sqrt{a}+1}+ar^2_1} + \sqrt{a}+1 \label{atomic-affine-overlap-3-7}\\
\phantom{a} &\leq& \frac{\sqrt{(\sqrt{a}+1)}}{2\sqrt{a}} + \sqrt{a} + 1 \leq \sqrt{a} + 2.
\end{eqnarray}
By dividing by $\ell_2$ we obtain (\ref{atomic-affine-overlap-3-2}), and we further have (\ref{atomic-affine-overlap-3-3}) due to the definition of atomic unsplittable. Due to the variational inequality, we have $\Phi_1(f_1,0) \leq \Phi_2(f_1,f_2)$ and $\Phi_2(0,f_2) \leq \ell\Phi_1(f_1,f_2)$ where $\ell = \ell_1/\ell_2$ and $\ell_1$ and $\ell_2$ represent the number of edges in path 1 and path 2 respectively. Inequality (\ref{atomic-affine-overlap-3-4}) is derived from the previous inequality by using the variational inequality.
Further, by dividing by $r_2 + c'_1$ and substituting $r_1$ with $r_2/\sqrt{a}$, we have (\ref{atomic-affine-overlap-3-5}).
In other words, $\frac{r_2+c'_1}{ar_1+r_2+c'_1} \geq \frac{r_2+c'_1}{\sqrt{a}r_2+r_2+c'_1} \geq \frac{1}{\sqrt{a} + 1}$.
Since $\frac{r_1(ar_1+c'_2)+r_2(r_2+c'_2)}{r_2\frac{r_2+c'_2}{\sqrt{a}+1}+r_1(ar_1+c'_2)} \leq (\sqrt{a} + 1)\frac{r_1(ar_1+c'_2)+r_2(r_2+c'_2)}{r_2(r_2+c'_2)+r_1(ar_1+c'_2)} = \sqrt{a} + 1$ due to $a \geq 1$, we have (\ref{atomic-affine-overlap-3-7}).
In (\ref{atomic-affine-overlap-3-7}), $r_1/r_2 = 1/\sqrt{a(\sqrt{a}+1)}$ and thus we have the second last inequality.
Since $a \geq 1$, $\frac{\sqrt{(\sqrt{a}+1)}}{2\sqrt{a}} \leq 1$ and our case analysis and proof is complete.

\section{Conclusion}\label{secCon}
In this paper we have studied the price of anarchy for $k$-commodity
non-atomic and atomic network flows with heterogeneous and decomposable delay functions.
We have also obtained improved  bounds on the price of anarchy for $2$-commodity
% * <byulpyo@gmail.com> 2018-04-26T02:58:31.181Z:
%
% ^.
% * <byulpyo@gmail.com> 2018-04-26T02:58:27.817Z:
%
% ^.
% * <byulpyo@gmail.com> 2018-04-26T02:58:18.873Z:
%
% ^.
atomic flows with heterogeneous uniform delay functions.

Further studies could include convex functions that model the
behavior of network queue delays more accurately, in order to shed more light on
the impact of differentiated services.

\section*{Acknowledgement}

We thank the referees for their valuable comments. This research was supported in part via NSF grant CCF 1451574.
We thank the anonymous referees for insightful and detailed comments and Mohit Hota for help with figures.

\end{document}